\documentclass[12pt]{article}
\usepackage{amsmath}
\usepackage{amssymb}
\usepackage{appendix}
\usepackage{subcaption}
\usepackage{mathtools}
\usepackage{graphicx,psfrag,epsf, color}
\usepackage{enumerate}
\usepackage{natbib}
\usepackage[ruled,vlined]{algorithm2e}

\usepackage{comment}
\usepackage{amsthm,ulem}
\usepackage{hyperref}

\usepackage{authblk}
\usepackage{multirow}

\newtheorem{theorem}{Theorem}[section]
\newtheorem{corollary}{Corollary}[theorem]
\newtheorem{lemma}[theorem]{Lemma}

\begin{document}

\def\spacingset#1{\renewcommand{\baselinestretch}%
{#1}\small\normalsize} \spacingset{1}

\newcommand{\bbeta}{\boldsymbol{\beta}}
\newcommand{\btau}{\boldsymbol{\tau}}
\newcommand{\bv}{\boldsymbol{v}}
\newcommand{\bX}{\boldsymbol{X}}
\newcommand{\bx}{\boldsymbol{x}}
\newcommand{\bY}{\boldsymbol{Y}}
\newcommand{\by}{\boldsymbol{y}}
\newcommand{\bD}{\boldsymbol{D}}
\newcommand{\bU}{\boldsymbol{U}}
\newcommand{\bW}{\boldsymbol{W}}
\newcommand{\bZ}{\boldsymbol{Z}}
\newcommand{\bz}{\boldsymbol{z}}
\newcommand{\bA}{\boldsymbol{A}}
\newcommand{\bepsilon}{\boldsymbol{\epsilon}}
\newcommand{\Lhalf}{L_{\frac{1}{2}}}
\newcommand{\tbeta}{\tilde{\beta}}
\newcommand{\tbbeta}{\tilde{\boldsymbol{\beta}}}
\newcommand{\bSigma}{\boldsymbol{\Sigma}}
\newcommand{\btheta}{\boldsymbol{\theta}}
\newcommand{\bLambda}{\boldsymbol{\Lambda}}
\newcommand{\bb}{\boldsymbol{b}}

\title{Bayesian $L_{\frac{1}{2^{}}}$ regression }
\author[1, 2]{Xiongwen Ke\thanks{
   Corresponding author: kexiongwen@csu.edu.cn}}
\author[2,3]{Yanan Fan}
\affil[1]{School of Mathematics and Statistics\\
     Central South University, China\\}
    \affil[2]{School of Mathematics and Statistics\\
    UNSW, Sydney 2052, Australia \\}
\affil[3]{Data61, CSIRO Sydney 2015, Australia}

\date{}
\maketitle

\begin{abstract}
It is well known that Bridge regression \cite{knight2000asymptotics} enjoys superior theoretical properties when compared to traditional LASSO. However, the current latent variable representation of its Bayesian counterpart, based on the exponential power prior, is computationally expensive in higher dimensions. In this paper, we show that the exponential power prior has a closed form scale mixture of normal decomposition for $\alpha=(\frac{1}{2})^\gamma, \gamma \in  \{1, 2,\ldots\}$. We call these types of priors $L_{\frac{1}{2}}$ prior for short. We develop an efficient partially collapsed Gibbs sampling scheme for computation using the $L_{\frac{1}{2}}$ prior and study theoretical properties when $p>n$.  In addition, we introduce a non-separable Bridge penalty function inspired by the fully Bayesian formulation and a novel, efficient coordinate descent algorithm.
We prove the algorithm's convergence and show that the local minimizer from our optimisation algorithm has an oracle property. Finally, simulation studies were carried out to illustrate the performance of the new algorithms. Supplementary materials for this article are available online.
\end{abstract}

\noindent%

\spacingset{1.45}
\section{Introduction}
Consider the linear regression problem 
\begin{equation}\label{eq:lm}
\bY=\bX \bbeta+\sigma\bepsilon,
\end{equation}
where $\bY=(y_1,\ldots, y_n)$ is 
an $n$-dimensional response vector, assumed to have been centred to 0 to avoid the need for an intercept. $\bX$ is an $n\times p$ design matrix consisting of $p$  {centered and standardised covariate measurements with $||X_j||^2=n, 1\leq j \leq p$}. $\bbeta=(\beta_1,\ldots, \beta_p)^T$ is a $p\times 1$ vector of unknown coefficients, $\bepsilon=(\epsilon_1,\ldots,\epsilon_n)$ is a vector of i.i.d noise with mean $0$ and variance 1. $\sigma>0$ is the error standard deviation. The Bridge estimator \citep{frank1993statistical} is the solution to the objective function of the form
$$
\arg \min_{\bbeta} \frac{1}{2\sigma^{2}}\|\bY-\bX \bbeta\|^{2}+\lambda \sum_{j=1}^{p}\left|\beta_{j}\right|^{\alpha}
$$
with $0<\alpha<1, \lambda >0$. 
When $\alpha=1$, the Bridge estimator is the same as the LASSO estimator, while the case of $\alpha=0$ is the NP-hard problem of subset selection, so the Bridge estimator is a compromise between the two.

Several papers have studied the statistical properties of Bridge regression estimators and their optimisation strategies from a frequentist point of view  \citep{knight2000asymptotics,huang2008asymptotic,zou2008one}. 
From a Bayesian perspective, the Bridge penalty induces the exponential power prior distribution for $\bbeta$ of the form, see \cite{polson2014bayesian}.
$$\pi(\bbeta|\lambda,\alpha) =\prod_{j=1}^p\frac{\lambda^{1 / \alpha}}{2 \Gamma(1+1 / \alpha)} \exp \left(-\lambda\left|\beta_{j}\right|^{\alpha}\right),\quad
 0<\alpha<1.$$ 
This prior can be expressed as a scale mixture of normals \citep{polson2014bayesian,armagan2009variational,west1987scale} 
$$\pi\left(\bbeta | \lambda, \alpha\right)\propto \prod_{j=1}^p\int_{0}^{\infty} \mathcal{N}\left(0, \tau_{j}^{2} \lambda^{-2 / \alpha}\right) \pi\left(\tau_{j}^{2}\right) d \tau_{j}^{2},$$ where $\pi(\tau_{j}^{2})$  is the density for scale parameter $\tau_{j}^{2}$. However, there is no closed form expression for $\pi(\tau_{j}^{2})$,
\cite{polson2014bayesian} proposed to work with the conditional distribution of $\btau^2|\bbeta$ in their MCMC strategy:
$$
\pi(\tau^{2} | \bbeta) = \prod_{j=1}^{p}\frac{\exp \left(-\lambda^{2 / \alpha}\beta_{j}^{2} /\tau_{j}^{2}\right) \pi\left(\tau_{j}^{2}\right)}{E_{\pi\left(\tau_{j}^{2}\right)}\left\{\exp \left(-\lambda^{2 / \alpha}\beta_{j}^{2} /\tau_{j}^{2}\right)\right\}},
$$
which is an exponentially tilted stable distribution. They suggested the use of the double rejection sampling algorithm of  \cite{devroye2009random} to sample this conditional distribution. However, the approach is complicated and not easy to scale to high dimensions. 
A second approach proposed by \cite{polson2014bayesian},
is based on a scale mixture of triangular (SMT) representation. 
More recently, an alternative
uniform-gamma representation was introduced by \cite{mallick2018bayesian}.

When the design matrix $\bX$ is orthogonal, both the scale mixture of triangles and the scale mixture of uniforms work well, but both of them suffer from poor mixing when the design matrix is strongly collinear. This is because both these two sampling schemes need to be generated from truncated multivariate normal distributions, which are still difficult to efficiently sample in higher dimensions.

In this paper, we propose a new data augmentation strategy, which provides us with a Laplace mixture representation of the exponential power prior at $\alpha=\left(\frac{1}{2}\right)^{^{\gamma}}, \gamma=\{1,2,\ldots\}$. For notational simplicity, we refer to these simply as $\Lhalf$
prior in the remainder
of this paper. 
This closed form representation introduces extra latent variables,  allowing us to circumvent sampling the difficult conditional distribution $\btau^{2}|\bbeta$. In addition, using the Laplace mixture representation, the resulting EM and LLA \citep{zou2008one} algorithms are identical. We further leverage the conjugacy structure in our model by removing certain conditional components in the Gibbs sampler without disturbing the stationary distribution. This is done by following the partially collapsed Gibbs sampling strategy \citep{park2009partially}. 

With the use of continuous shrinkage priors, the optimal Bayes point estimates under many common loss functions (e.g. $L_{1}$ and $L_{2}$) will not have exact zeros, requiring extra and ad hoc steps for variable selection. 
It is well known that penalized likelihood estimators have a Bayesian interpretation as posterior modes under the corresponding prior. Thus, it is possible to construct very flexible sparse penalty functions from a Bayesian perspective, whose solution is the generalized thresholding rule \citep{rovckova2016bayesian}. We consider a full Bayesian formulation of the exponential power prior, which introduces a non-separable Bridge (NSB) penalty
\begin{equation}\label{margpen}
{\bf pen}(\bbeta)\coloneqq -\log \int \prod_{j}\pi(\beta_{j}|\lambda)\pi(\lambda)d\lambda
\end{equation}
with the hyperparameter $\lambda$  integrated out with respect to some suitable choice of hyperprior. \cite{polson2012local} and \cite{song2017nearly} showed that if the global-local shrinkage prior has a heavy and flat tail, and allocates a sufficiently large probability mass in a very small neighbourhood of zero, then its posterior properties are as good as those of the spike-and-slab prior. However, these properties often lead to unbounded derivatives for the penalty function at zero. In particular, the proposed NSB penalty is a non-convex, non-separable and non-Lipschitz function, posing a very challenging problem for optimisation. For the second main innovation of this paper, we propose a simple coordinate descent (CD) algorithm which is capable of producing sparse solutions.
We further show that the convergence of the algorithm to the local minima is guaranteed theoretically. 

The paper is structured as follows: In Section \ref{sec:bridgeprior}, we present the normal mixture representation for $\Lhalf$ prior. 
In Section \ref{sec:postcons}, we produce theoretical conditions for posterior consistency and contraction rates. In Section \ref{sec:pcg}, we propose an efficient posterior sampling scheme based on the partially collapsed Gibbs sampler. In Section \ref{sec:VS}, we develop a coordinate descent optimisation algorithm for variable selection and study its oracle properties.
In Section \ref{sec:sim}, we provide some numerical results. Finally in Section \ref{sec:conclusion}, we conclude with some discussions.

\section{The $\Lhalf$ prior}\label{sec:bridgeprior}

In this section, we begin with the decomposition of the $\Lhalf$ prior, with $\gamma \in \{1,2,3,\ldots\}$.

\subsection{The Laplace mixture representation}
The scale mixture representation was first found by \cite{west1987scale}, 
 who showed that a function $f(x)$ is completely monotone if and only if it can be represented as a Laplace transform of some 
function $g(\cdot)$:
$$f(x)=\int_{0}^{\infty} \exp (-s x)g(s) \mathrm{d}s.$$
To represent the exponential power prior with $0<\alpha<1$ as a Gaussian mixture, we let $x=\frac{t^{2}}{2}$,
$$\exp \left(-|t|^{\alpha}\right)=\int_{0}^{\infty} \exp \left(-s t^{2} / 2\right) g(s) \mathrm{d} s, $$
where $\exp \left(-|t|^{\alpha}\right)$ is the the Laplace transform of $g(s)$ evaluated at $\frac{t^{2}}{2}$ \citep{polson2012local,polson2014bayesian}. Unfortunately, there is no general closed form expression for $g(s)$. However, for the special case $\alpha=(\frac{1}{2})^{\gamma}$ where $\gamma \in \{1,2,3 \ldots \}$, we can construct a data augmentation scheme to represent $g(s)$. We begin with the following lemma.

\begin{lemma}\label{lemma2.2}
The exponential power distribution of the form 
$
\pi(\beta)= \frac{\lambda^{2^{\gamma}}}{2 (2^{\gamma}!)} \exp \left(-\lambda\left|\beta\right|^{\frac{1}{2^{\gamma}}}\right),
$ with $\alpha=(\frac{1}{2})^{\gamma}$ with $\gamma \in \{1,2,3\ldots\}$, 
can be decomposed as

 $$\pi(\beta|s)= \frac{1}{2(2^{\gamma-1}!)s^{2^{\gamma-1}}}\exp\left(-\frac{|\beta|^{\frac{1}{2^{\gamma-1}}}}{s}\right),
\quad  s\sim{\mathrm{Gamma}\left(\frac{2^{\gamma}+1}{2},\frac{\lambda^{2}}{4}\right)}$$
or equivalently
$$ \pi(\beta|v,\lambda)= \frac{\lambda^{2^{\gamma}}}{2(2^{\gamma-1}!)v^{2^{\gamma-1}}}\exp \left(-\frac{\lambda^{2}|\beta|^{\frac{1}{2^{\gamma-1}}}}{v}\right), \quad
 v \sim \mathrm{Gamma}\left(\frac{2^{\gamma}+1}{2},\frac{1}{4}\right).$$
\end{lemma}

Using Lemma \ref{lemma2.2}, the main Theorem \ref{thm:2.3} provides the analytic expressions under the scale mixture of normal representation.

\begin{theorem}\label{thm:2.3}
The exponential power prior with $\alpha=(\frac{1}{2})^{\gamma}$ where $\gamma \in \{1,2,3\ldots\}$ can be represented as Laplace mixture with $i=1,\ldots,\gamma-1$,
$$ \beta|v_{1},\lambda \sim \mathrm{DE}\left(\frac{v_{1}}{\lambda^{2^{\gamma}}}\right), \quad
 v_{i}|v_{i+1} \sim \mathrm{Gamma}\left(\frac{2^{i}+1}{2},\frac{1}{4v_{i+1}^{2}}\right), \quad
 v_{\gamma} \sim \mathrm{Gamma}\left(\frac{2^{\gamma}+1}{2},\frac{1}{4}\right),\\$$
which can be represented as global-local shrinkage prior,
\begin{equation}\label{eq:glprior}
\begin{aligned}
&\beta |\tau^{2}, \lambda  \sim N\left(\mathbf{0},  \frac{\tau^{2}}{\lambda^{2^{\gamma+1}}}\right),\quad 
&\tau^{2}|v_{1} \sim  \mathrm{Exp}\left(\frac{1}{2v_{1}^{2}}\right), \\ 
& v_{i}|v_{i+1} \sim \mathrm{Gamma}\left(\frac{2^{i}+1}{2},\frac{1}{4v_{i+1}^{2}}\right), \quad 
& v_{\gamma} \sim \mathrm{Gamma}\left(\frac{2^{\gamma}+1}{2},\frac{1}{4}\right),\\
\end{aligned}  
\end{equation}
\noindent where  $\tau$ is the local shrinkage parameter and $\frac{1}{\lambda^{2^{\gamma}}}$ is the global shrinkage parameter. In addition, when $\gamma=1$, the terms $v_i|v_{i+1}$ vanish. 
\end{theorem}
The proof can be found in Section 1 of the Appendix, together with some intuitions on why the $\Lhalf$ prior is a suitable shrinkage prior. As we can see, 
introducing latent variables $v_i$ to augment the distribution for the local shrinkage parameter $\tau$,  leads to a computationally efficient decomposition. The augmented representation now makes it easier to create computationally efficient algorithms such as MCMC or the EM algorithm.

\subsection{Choice of hyperprior}\label{sec:hyperprior}
The hyperparameter $\lambda$ controls global shrinkage and is critical to the success of the $L_{\frac{1}{2}}$ prior. It is possible 
to fix $\lambda$ by selecting a value for it via an empirical Bayes approach \citep{casella2001empirical}, however in extremely sparse cases, these estimates can collapse to 0 
\citep{scott2010bayes,datta2013asymptotic}. We therefore assign a hyperprior to $\lambda$, allowing the model to achieve a level of self-adaptivity.  Here we use
\begin{equation}\label{eq:lambdaprior}
\lambda \sim \mathrm{Gamma}(a,1/b),
\end{equation}
where the parameters $a$ and $b$ are user specified. In practice, we have found that in $n>p$ cases, the algorithm is insensitive to the values of $a$ and $b$. In $n<p$ cases, our experience for several different scenarios suggests that setting $a$ to around $0.1p$ and $b=1$ works well for $\gamma=1$ and 
setting $a=1$ and $b$ to around $0.2p$ works well for $\gamma=2$.

\section{Posterior consistency and  contraction rate}\label{sec:postcons}
We focus on the high dimensional  
case where the number of predictors $p$ is much larger than the number of observations $n$ ($p > n$) and most of the coefficients in the parameter vector $\bbeta=(\beta_1,\ldots, \beta_p)$ are zero.
We consider the $\Lhalf$ prior  
on $\bbeta$ under the following setup,
\begin{equation*}\label{eq:slr}
\bY | \bbeta, \sigma^{2}  \sim \mathcal{N}_{n}\left(\bX\bbeta, \sigma^{2} \mathbf{I}_{\mathbf{n}}\right), \quad
 \pi(\bbeta|\lambda) \propto  \exp \left(-\lambda\sum_{j=1}^{p}\left|\beta_{j}\right|^{\frac{1}{2^{\gamma}}}\right), \quad
 \pi\left(\sigma^{2}\right) \propto \sigma^{-2}.
\end{equation*}

\cite{mallick2018bayesian} showed strong consistency of the posterior under exponential power prior in the case $p < n$.
When $p  > n$, the non-invertibility of the design matrix complicates analysis. In the following, we show that under the local invertibility assumption of the Gram matrix $\frac{\bX^{T}\bX}{n}$, the $L_{2}$ contraction rates for the posterior of $\bbeta$ are nearly optimal and no worse than the convergence rates achieved
by the spike-and-slab prior \citep{castillo2015bayesian}. Therefore, there is no estimation performance loss due to switching from the spike-and-slab prior to the $L_{\frac{1}{2}}$ prior.

In what follows, we rewrite the dimension $p$ of the model by $p_{n}$ to indicate that the number of predictors can increase with the sample size $n$, and similarly, we rewrite $\lambda$ as $\lambda_n$. We use $\bbeta_{0}$ and $\sigma_{0}$ to indicate the true regression coefficients and the true standard deviation. Let $S_{0}$ be the set containing the indices of the true nonzero coefficients, where $S_{0} \subseteq \left\{1,..,p_{n}\right\}$, then $s_{0}=|S_{0}|$ denotes the true size of the model. We use $\Pi_{n}(\cdot)$ and $\Pi_{n}(\cdot \mid \boldsymbol{Y})$ to emphasise that both the prior and posterior distributions are sample size dependent. We let $P_{\btheta_{0}}^{n}$  denote the probability measure under the truth $\btheta_0=(\bbeta_0, \sigma_0)$. We now state our main theorem on the posterior contraction
rates for exponential power prior based on the following assumptions:\\
{\bf A1.} The dimension is high,  $p_{n} \succ n$ and $\log(p_{n})=o(n)$.\\
{\bf A2.} The true number of nonzeros in $\bbeta_{0}$ satisfies $s_{0}=o(n/\log p_{n})$. \\
{\bf A3.} All the covariates are uniformly bounded. In other words, there exists a constant $k>0$ such that $\lambda_{\max}(\bX^{T}\bX)\leq{kn^{\alpha}}$ for some  $\alpha \in{[1,+\infty)}$. \\
{\bf A4.} There exist constants $v_{2}>v_{1}>0$ and an integer $\widetilde{p}$ satisfying $s_{0}=o(\widetilde{p})$ and $\widetilde{p}=o(s_{0}\log n)$, so that $nv_{1}\leq{\lambda_{min}(\bX_{s}^{T}\bX_{s})}\leq{\lambda_{max}(\bX_{s}^{T}\bX_{s})}\leq{nv_{2}}$ for any model of size $|s|\leq{\widetilde{p}}$.\\
{\bf A5.} $||\bbeta_{0}||_{\infty}={E}$ and $E$ is some positive number independent of $n$.

Assumption (A1) allows the number of covariates $p$ to grow at a nearly exponential rate with sample size $n$. Assumption (A2) specifies the growth rate for the true model size $s_{0}$. Assumption (A3) bounds the eigenvalues of $\bX^{T}\bX$ from above and is less stringent than requiring all the eigenvalues of the Gram matrix $(\bX^{T}\bX/n)$ to be bounded away from infinity. Assumption (A4) ensures that $\bX^{T}\bX$ is locally invertible over a small subspace. Finally, Assumption (A5) requires boundedness of the maximum signal size for the true $\bbeta$.

\begin{theorem}\label{thm:PC}
(Posterior contraction rates) Let $\epsilon_{n}=\sqrt{s_{0}\log p_{n}/n}$ and suppose that assumptions (A1)-(A5) hold. Under the linear regression model with unknown Gaussian noise, we endow $\bbeta$ with the exponential power prior $\Pi_{n}(\bbeta)=\frac{\lambda_{n}^{\frac{1}{\alpha}}}{2\Gamma(1+\frac{1}{\alpha})}\exp{\left(-\lambda_{n}|\bbeta|^{\alpha}\right)}$,
where $0<\alpha\leq\frac{1}{2}$ and $\lambda_{n} \propto (\log p_{n})$ and $\sigma^{2} \sim \mathrm{InvGamma}(\delta,\delta)$. Then 
$$
\begin{aligned}
& \Pi_{n}\left(\boldsymbol{\beta}:\left\|\boldsymbol{\beta}-\boldsymbol{\beta}_{0}\right\|_{2} \gtrsim  \epsilon_{n} \mid \boldsymbol{Y}\right) \rightarrow 0 \quad \text{a.s.} \,\, P_{\theta_{0}}^{n}  \,\,  \text{as} \,\, n \rightarrow  \infty,\\
& \Pi_{n}\left(\boldsymbol{\beta}:\left\|\boldsymbol{X} \boldsymbol{\beta}-\boldsymbol{X} \boldsymbol{\beta}_{0}\right\|_{2} \gtrsim  \sqrt{n} \epsilon_{n} \mid \boldsymbol{Y}\right) \rightarrow 0 \quad \text{a.s.} \,\, P_{\theta_{0}}^{n}  \,\,  \text{as} \,\, n \rightarrow  \infty,\\
&\Pi_{n}\left(\sigma^{2}:\left|\sigma-\sigma_{0}\right|\gtrsim  \sigma_{0}\epsilon_{n} \mid \bY\right) \rightarrow 0 \quad \text{a.s.} \,\, P_{\theta_{0}}^{n}  \,\,  \text{as} \,\, n \rightarrow  \infty,
\end{aligned}
$$
where $\Pi_{n}(\cdot \mid \boldsymbol{Y})$ denotes the posterior distribution under the prior $\Pi_{n}$.
\end{theorem}
\begin{theorem}\label{thm:dim}
(Dimensionality) We define the generalized dimension as
$$\gamma_{j}(\beta_{j})=I(|\beta_{j}|>a_{n}) \, \, \, \text{and} \, \, \, |\gamma(\bbeta)|=\sum_{j=1}^{p_{n}}\gamma_{j}(\beta_{j}),$$
where $a_{n}=\frac{\sigma_{0}}{\sqrt{2k}}\sqrt{s_{0}\log p_{n}/n}/p_{n}$ and $\lim_{n\rightarrow{\infty}}a_{n}=0$. Then under assumptions (A1-A5), for sufficiently large $M_{3}>0$, we have
$$\sup _{\theta_{0}} P_{\theta_{0}}^{n} \Pi\left(\boldsymbol{\beta}:|\mathcal{\gamma}(\boldsymbol{\beta})|>M_{3} s_{0} \mid \boldsymbol{Y}\right) \rightarrow 0.$$
\end{theorem}
Detailed proofs for Theorems \ref{thm:PC} and \ref{thm:dim} can be found in Section 5 of the Appendix.

\section{MCMC for Bayesian $\Lhalf$ regression}\label{sec:pcg}

For the linear model with Gaussian error of the form  \eqref{eq:lm}, together with the $\Lhalf$ prior of the form \eqref{eq:glprior}, and the hyperprior \eqref{eq:lambdaprior},  standard implementation of the Gibbs sampler is not efficient as the full 
conditional distributions for $\bv_{1}$,...,$\bv_{\gamma}$ and $\lambda$ do not follow any standard distributions. Additional sampling strategies, such as the Metropolised-Gibbs sampler or adaptive rejection sampling \citep{gilksWilks92} would be needed. The partially collapsed Gibbs sampler (PCG) \citep{park2009partially} speeds up convergence by partially marginalising some of the conditional components from its parent Gibbs sampler, which also provides closed form conditional posteriors in our case. \cite{van2008partially} showed that, by applying three basic steps: marginalize, permute and trim, one can construct the PCG sampler which maintains the target joint posterior distribution as the stationary distribution. 

Following the prescription in \cite{van2008partially}, the first step of the PCG sampler is marginalisation,  where we sample $\boldsymbol{\tau}^{2}$,   $\boldsymbol{v}_{1}$,...,$\boldsymbol{v}_{\gamma}$ multiple times in the Gibbs cycle without perturbing the stationary distribution. In the second step, we permute the ordering of the update which helps us identify trimming strategies. Based on the final trimming step, we are able to construct a sequential sampling scheme to sample the posterior where each update is obtained from a closed form distribution with some of the conditional components being removed.
Details of the three steps and derivations for the conditional distributions are given in Section 2 of the Appendix. Note that the PCG sampler for $\Lhalf$ prior can also be extended to normal mixture likelihoods (See Section 4 of the Appendix for extensions to logistic and quantile regressions). 
We provide the sampling scheme for the model in Equation \eqref{eq:lm} in Algorithm 1.

\begin{algorithm}[h]
Initialise the parameters $\bbeta, \lambda, \bv,\btau^2, \sigma^2$  and set $t=1$, sequentially update the parameters until at $t=T$, via
\begin{equation*}
\begin{aligned}
{\bf S1.}   \quad &  \mbox{Sample }  \bbeta \mid \tau^{2},\lambda,\sigma^{2} 
\sim \mathrm N_{p}\left((\bX^{T}\bX+\sigma^{2}\lambda^{2^{\gamma+1}}\bD_{\btau^{2}}^{-1})^{-1}\bX^{T}\bY,\sigma^{2}(\bX^{T}\bX+\sigma^{2}\lambda^{2^{\gamma+1}}\bD_{\tau^{2}}^{-1})^{-1} \right)\\
{\bf S2.}   \quad &  \mbox{Sample } \lambda \mid \bbeta
\sim \mathrm{Gamma}\left(2^{\gamma}p+a,\sum_{j=1}^{p}\left|\beta_{j}\right|^{\frac{1}{2^{\gamma}}}+\frac{1}{b}\right)\\
{\bf S3.} \quad &  \mbox{Sample } \frac{1}{{v}_{\gamma, j}} \mid \beta,\lambda \sim \mathrm{InvGaussian}\left (\frac{1}{2\lambda|\beta_{j}|^{\frac{1}{2^{\gamma}}}},\frac{1}{2}\right), \quad j=1,...,p\\
& \mbox{\bf{If}} \,\, \gamma \geq 2: \\
&\quad \mbox{\bf{For}}\,\, i=\gamma-1 \,\, \text{to} \,\, 1: \\
& \quad \quad \mbox{Sample } \frac{1}{{v}_{i,j}} \mid \beta, \lambda, v_{i+1,j},  \sim \mathrm{InvGaussian}\left (\frac{1}{2v_{i+1,j}\lambda|\beta_{j}|^{\frac{1}{2^{i}}}},\frac{1}{2{v}_{i+1,j}^{2}}\right ),\quad j=1,...,p\\ 
& \quad \mbox{\bf{End for}}\\
& \mbox{{\bf End if}}\\
{\bf S4.}   \quad &  \mbox{Sample} \,\, \frac{1}{{\tau}_{j}^{2}} \mid \beta,\lambda,v_{1,j}  \sim \mathrm{InvGaussian}\left (\frac{1}{{\lambda}^{2^{\gamma}}{v}_{1,j}|\beta_{j}|},\frac{1}{{v}_{1,j}^{2}}\right ), \quad j=1,...,p\\
{\bf S5.}   \quad &  \mbox{Sample }  {\sigma}^{2} \mid \bbeta
\sim  \mathrm{InvGamma}\left(\frac{n}{2},\frac{1}{2}(\bY-\bX\bbeta)^{T}(\bY-\bX\bbeta)\right )\\
\end{aligned}
\end{equation*}
\caption{PCG $\Lhalf, \gamma \in \{1,2,3,\ldots\}$ update algorithm}\label{alg:PCG2}
\end{algorithm}

\section{Variable selection: posterior mode search via coordinate descent optimisation}\label{sec:VS}
Whilst the $L_{\frac{1}{2}}$ prior presented in Section \ref{sec:bridgeprior} - \ref{sec:pcg} produces full posterior distributions for all the parameters, the posterior probability measure of $\beta_{j}=0$ is always zero, making variable selection difficult. One approach may be to check whether the posterior credible interval contains zero or not at a nominal level, say, $95\%$ \citep{tadesse2021handbook}.  However, it is not clear how the level can be chosen optimally. If the goal of inference is to identify important variables very quickly, we propose below a fast optimisation strategy which searches for the posterior modes and is capable of producing sparse solutions suitable for variable selection.

Coordinate descent types of algorithms have been developed for Bridge penalties previously, see for example, \cite{marjanovic2012l_q,marjanovic2013exact,marjanovic2014l_}. However, the Bridge penalty is limited by its lack of ability to adapt to the sparsity pattern across the coordinates. 
Without knowing the true sparsity, specifying the global shrinkage parameter $\lambda$ can be challenging. In the spirit of the full Bayesian treatment, assigning a suitable hyperprior for $\lambda$ can also achieve a degree of self-adaptiveness for the optimisation algorithm, see also \cite{rovckova2016bayesian,rovckova2018spike}. With this motivation, we first integrate out $\lambda$ with respective to the hyperprior in (\ref{eq:lambdaprior}) to obtain the marginalised prior $\pi(\bbeta)$.  Then together with the Jeffreys’ prior $\pi\left(\sigma^{2}\right) \propto \sigma^{-2}$, the marginalised log  posterior distribution can be written as:
\begin{equation}\label{eq:objective}
\begin{aligned}
\log \pi\left(\bbeta, \sigma^{2}\right |\bY) = & -\frac{1}{2 \sigma^{2}}\left\|\bY-\bX\bbeta\right\|_{2}^{2}-(n+2) \log \sigma+\log \int \pi(\bbeta|\lambda)\pi(\lambda)d\lambda+C\\
= & -\frac{1}{2 \sigma^{2}}\left\|\bY-\bX\bbeta\right\|_{2}^{2}-(n+2) \log \sigma+\log \pi(\bbeta)+C\\
=& -\frac{1}{2 \sigma^{2}}\left\|\bY-\bX\bbeta\right\|_{2}^{2}-(n+2) \log \sigma-\mathbf{pen}(\bbeta)+C',
\end{aligned}
\end{equation} 
where $C$ and $C'$ are constant terms and $\mathbf{pen}(\bbeta)=(2^{\gamma}p+a)\log \left(\sum_{j=1}^{p}|\beta_{j}|^{\frac{1}{2^{\gamma}}}+1/b\right)$
is the non-separable Bridge penalty (NSB).

When $\sigma^{2}$ is unknown, the non-convexity of the Gaussian negative log-likelihood is a well known difficulty for optimisation \citep{buhlmann2011statistics}. For our model,  optimising the objective function in Equation (\ref{eq:objective}) appears to work well in low dimensional cases $n>p$, but often shrinks everything to zero when $n<p$. It is not clear why variance estimation can sometimes fail in high dimensions. Here we suggest the following simple strategy to tackle this issue.

Suppose $\sigma_{0}^{2}$ is an unknown constant, under  the Bridge penalty, the objective function can be rewritten as 
\begin{equation*}
\arg \min_{\bbeta} \frac{1}{2}\|\bY-\bX \bbeta\|^{2}+\sigma_{0}^{2}\lambda \sum_{j=1}^{p}|\beta_{j}|^{\frac{1}{2^{\gamma}}}.
\end{equation*}
If we define $\lambda^{'}=\sigma_{0}^{2}\lambda$ as a new penalty parameter,  and assign the gamma prior in (\ref{eq:lambdaprior}) to $\lambda^{'}$ and integrate it out, we end up with an alternative objective function to \eqref{eq:objective} with the same penalty,
\begin{equation}\label{eq:objectiveS} 
 \log \pi\left(\bbeta\right |\bY)= -\frac{1}{2}\left\|\bY-\bX\bbeta\right\|_{2}^{2}-\mathbf{pen}(\bbeta)+C''. \\
\end{equation}
where $C''$ is constant term. We then propose a two stage approach. In stage one,
we find $\hat{\bbeta}$ and $\hat{s}$ (the number of nonzero elements in $\hat{\bbeta}$) by maximising \eqref{eq:objectiveS}. In stage two, we compute the variance using the estimator:
\begin{equation}\label{eq:sigma2}
\hat{\sigma}^{2}=\frac{\|\bY-\bX\hat{\bbeta}\|_{2}^{2}}{n-\hat{s}}.
\end{equation}
\cite{fan2011nonconcave} studied the oracle properties of the non-concave penalized likelihood estimator in a high dimensional setting with the two stage approach. \cite{fan2012variance}  showed that the variance can be consistently estimated by 
Equation (\ref{eq:sigma2}). 
The main difference in our approach here is that rather than using cross validation to determine $\lambda'$, we marginalise over $\lambda'$.

\subsection{The KKT condition and coordinate descent algorithm}\label{sec:CDalg}
For the optimisation of $\bbeta$, \cite{rovckova2016bayesian} gave the following necessary Karush-Kuhn-Tucker (KKT) condition for the global mode:
\begin{equation}\label{eq:rokova}
\beta_{j}=(\bX_{j}^{T}\bX_{j})^{-1}\left(|z_{j}|- \frac{\partial \operatorname{\bf pen}(\bbeta) }{\partial |\beta_{j}|}\right)_{+}\mathrm{sign}(z_{j}),
\end{equation}
where $z_{j}=\bX_{j}^{T}(\bY-\bX_{-j}\beta_{-j})$ and $(z)_+ = \max(0,z)$. The derivatives of $\operatorname{\bf pen}(\bbeta)$ play a crucial role here. For example, for fixed $\lambda$ under the Laplace prior, the solution is the LASSO estimator $\beta_{j}=(\bX_{j}^{T}\bX_{j})^{-1}\left(|z_{j}|-\sigma^{2}\lambda\right)_{+}\mathrm{sign}(z_{j})$. 
For NSB, we integrate out the hyperparameter $\lambda$ with respect to the gamma prior in (\ref{eq:lambdaprior}), which gives us
\begin{equation*}
\begin{aligned}
\pi(\beta_{1},...,\beta_{p}) & =\frac{(2^{\gamma}p+a-1)!}{\sqrt{\pi b}2^{p}[(2^{\gamma})!]^{p}} \frac{1}{(\sum_{j=1}^{p}|\beta_{j}|^{\frac{1}{2^{\gamma}}}+1/b)^{2^{\gamma}p+a}}.
\end{aligned}
\end{equation*}
Then
$$
\frac{\partial \operatorname{\bf pen}(\bbeta) }{\partial |\beta_{j}|} =- \frac{\partial \log \pi(\bbeta)}{\partial |\beta_{j}|}
= \frac{(2^{\gamma}p+a)|\beta_{j}|^{\frac{1}{2^{\gamma}}-1}}{2^{\gamma}(\sum_{j=1}^{p}|\beta_{j}|^{\frac{1}{2^{\gamma}}}+\frac{1}{b})}
=\left(p+\frac{a}{2^{\gamma}}\right)\frac{1}{|\beta_{j}|^{1-\frac{1}{2^{\gamma}}}}\frac{1}{\sum_{j=1}^{p}|\beta_{j}|^{\frac{1}{2^{\gamma}}}+\frac{1}{b}},
$$
and thus we have
\begin{equation}\label{eq:solKKT}
\beta_{j}=(\bX_{j}^{T}\bX_{j})^{-1}\left(|z_{j}|-\frac{C_{1}}{|\beta_{j}|+C_{2}|\beta_{j}|^{1-\frac{1}{2^{\gamma}}}}\right)_{+}\mathrm{sign}(z_{j}),
\end{equation}
where we set
\begin{equation}\label{eq:const}
C_{1}=\left(p+\frac{a}{2^{\gamma}}\right), \quad C_{2}=\sum_{i \neq j}|\beta_{i}|^{\frac{1}{2^{\gamma}}}+\frac{1}{b}. 
\end{equation}
For $\gamma>0$, the shrinkage term $\frac{\partial \operatorname{\bf pen}(\bbeta) }{\partial |\beta_{j}|}$ is global-local adaptive, it depends on the information both from itself and from the other coordinates \citep{polson2012local}. While for $\gamma=0$, $\frac{\partial \operatorname{\bf pen}(\bbeta) }{\partial |\beta_{j}|}=\frac{1}{\sum_{j=1}^{p} |\beta_{j}|+1/b}$, which does not have the local adaptive properties.

However, the KKT condition given by Equation \eqref{eq:rokova} does not apply to penalties with unbounded derivatives. In fact, the set of $\hat{\bbeta}$ satisfying Equation \eqref{eq:rokova} is only a superset of the solution set for the NSB penalty. We can see that for $\gamma>0$, the solution $\beta_j=0$ always satisfies Equation \eqref{eq:solKKT}. Here, we derive the true KKT condition for our penalty, which serves as a basis for the CD algorithm.

\begin{theorem}\label{thm:betasol}
Let the constants $\tbeta_{j}$ and $h(\tbeta_j)$, be such that they satisfy the following conditions
\begin{equation}\label{eq:betasol}
\begin{aligned}
1 & =(\bX_{j}^{T}\bX_{j})^{-1}C_{1}\frac{1+(1-\frac{1}{2^{\gamma}})C_{2}\tbeta_{j}^{-\frac{1}{2^{\gamma}}}}{\tbeta_{j}^{2}\left(1+C_{2}\tbeta_{j}^{-\frac{1}{2^{\gamma}}}\right)^{2}}\\
h(\tbeta_{j}) & =\tbeta_{j}+(\bX_{j}^{T}\bX_{j})^{-1}\frac{C_{1}}{\tbeta_{j}+C_{2}\tbeta_{j}^{(1-\frac{1}{2^{\gamma}})}}.
\end{aligned}
\end{equation}
Then the solutions satisfy 
\begin{equation}\label{eq:KKT}
\beta_{j}=T(\bbeta_{-j},\beta_{j})=\left\{\begin{array}{ll}
0 & \text{ if } (\bX_{j}^{T}\bX_{j})^{-1}|z_{j}|<h(\tbeta_{j})\\
0 & \text{ if } (\bX_{j}^{T}\bX_{j})^{-1}|z_{j}| \geq h(\tbeta_{j}) \quad \text{and} \quad \delta_{-j}(\beta_{j}^{''})>0\\
\operatorname{sign}(z_{j}) \beta_{j}^{''} & \text{ if } (\bX_{j}^{T}\bX_{j})^{-1}|z_{j}| \geq h(\tbeta_{j}) \quad \text{and} \quad \delta_{-j}(\beta_{j}^{''}) \leq 0,
\end{array}\right.
\end{equation}
where  $\delta_{-j}(\beta_{j})=L_{-j}(\beta_{j})-L_{-j}(0)$, $\beta_{j}^{''} \in [\tbeta_{j},(\bX_{j}^{T}\bX_{j})^{-1}|z_{j}|]$, $C_1$ and $C_2$ are as given in \eqref{eq:const} and $L_{-j}$ denotes the loss function in Equation (\ref{eq:objectiveS}) with all except the $j^{th}$ $\bbeta$ fixed.
\end{theorem}

The proof of the theorem  
is provided in Section 3 of the Appendix, where in steps one and two of the proof, we also get the following corollary: \begin{corollary}\label{corr:fixed_point}
$(\bX_{j}^{T}\bX_{j})^{-1}|z_{j}| \geq h(\tbeta_{j})$ if and only if 
$\beta_{j}^{''}$ exists and can be computed by fixed point iteration:
$\beta_{j}^{(k+1)}=\rho(\beta_{j}^{(k)})$, where
\begin{equation}\label{eq:fixpt}
\rho(\beta_{j})=(\bX_{j}^{T}\bX_{j})^{-1}|z_{j}|-(\bX_{j}^{T}\bX_{j})^{-1}\frac{C_{1}}{\beta_{j}+C_{2}\beta_{j}^{1-\frac{1}{2^{\gamma}}}}
\end{equation}
with the initial condition $\beta_{j}^{(0)}\in [\tbeta_{j},(\bX_{j}^{T}\bX_{j})^{-1}|z_{j}|]$. When $(\bX_{j}^{T}\bX_{j})^{-1}|z_{j}|=h(\tbeta_{j})$, $\beta_{j}^{''}=\tbeta_{j}$.
\end{corollary}

\noindent {\bf Remark 1:} From Theorem \ref{thm:betasol},  we see that the global mode is a blend of hard thresholding and non-linear shrinkage.
$h(\tbeta_{j})$ is the selection threshold where $\tbeta_{j}$ needs to be solved numerically according to Equation \eqref{eq:betasol}. Rather than evaluating the selection threshold by solving this difficult non-linear equation directly, we suggest selecting the variable by checking the convergence of the fixed point iterations. By corollary \ref{corr:fixed_point}, we can run the fixed point iteration with an initial guess $\beta_{j}^{(0)}=(\bX_{j}^{T}\bX_{j})^{-1}|z_{j}|$.  If it fails to converge after a long iteration or $\rho(\beta_{j}^{(k)})<0$, we conclude that $\beta_{j}=0$.

However, this is still computationally intensive. A single block of update requires running the non-parallel fixed point iteration algorithms $p$ times sequentially. The next lemma allows us to run the fixed point iteration only on a small subset of variables. It is also a useful auxiliary result for us to show the oracle property of the estimator later.

\begin{lemma}\label{lem:auxiliary}
If the regressors have been
centered and standardized with $\|\bX_{j}\|^{2}=n$, for $1\leq j \leq p$, then by Equation \eqref{eq:betasol}, for $\gamma \in \{1,2,3,\ldots\}$, we have
$2\tbeta_{j}\leq h(\tbeta_{j}) \leq 3\tbeta_{j}$ and 
$$\tbeta_{j}^{(2-\frac{1}{2^{\gamma}})}=\frac{C_{1}M'(\gamma)}{\bX_{j}^{T}\bX_{j}}\frac{1}{\tbeta_{j}^{\frac{1}{2^{\gamma}}}+C_{2}}=M'(\gamma)\frac{p+\frac{a}{2^{\gamma}}}{n\tbeta_{j}^{\frac{1}{2^{\gamma}}}+nC_{2}}$$
for $M'(\gamma)\in (\frac{1}{2},1]$.
If the estimator satisfies $\sum_{i=1}^{p_{n}} |\hat{\beta}_{j}|^{\frac{1}{2^{\gamma}}}=O(s_{0})$ and $\|\hat{\beta}\|_{\infty}< \infty$, then $\tbeta_{j}^{(2-\frac{1}{2^{\gamma}})}=O(\frac{p_{n}}{s_{0}n+nb^{-1}})$. In addition, there exists another lower bound for the selection threshold $u(\beta_{-j}) < h(\tbeta_{j})$, which is not a function of $\tbeta_{j}$.
\begin{equation}\label{eq:lowbound}
u(\beta_{-j})=2\left\{\frac{C_{1}(\bX_{j}^{T}\bX_{j})^{-1}}{2C_{2}+2[(\bX_{j}^{T}\bX_{j})^{-1}|z_{j}|]^{\frac{1}{2^{\gamma}}}}\right\}^{\frac{1}{2-\frac{1}{2^{\gamma}}}}.
\end{equation}

\end{lemma}
The proof of Lemma \ref{lem:auxiliary} is in Section 3 of the Appendix.\\
\noindent {\bf Remark 2:} Lemma \ref{lem:auxiliary} provides an explicit lower bound for the selection threshold $h(\tbeta_{j})$, we see that $\hat{\beta}_{j}=0$ if $(\bX_{j}^{T}\bX_{j})^{-1}|z_{j}| \leq u(\beta_{-j})$. Therefore, we only need to run the fixed point iterations in $(\bX_{j}^{T}\bX_{j})^{-1}|z_{j}|> u(\beta_{-j})$ cases. To further speed up, we suggest using the null model $\bbeta_{Initial}=0$ as initialisation. Then in the early stages of the updates, the lower bound (\ref{eq:lowbound}) will be large, which means we only need to run fixed point iteration in a very small set of variables.\\

We denote $\bbeta_{-j}^{i}=(\beta_{1}^{(i+1)},...,\beta_{j-1}^{(i+1)},\beta_{j+1}^{(i)},...,\beta_{p}^{(i)})$ as the vector of regression coefficients with the $j^{th}$ element left out.  Then with some small error tolerance  $\epsilon$, terminal number of fixed point iterations $T$, and user-specified values for the hyperparameters $a, b$, the CD optimisation algorithm for NSB penalty is given in Algorithm \ref{alg:VS}.\\

\begin{algorithm}[H]
\SetAlgoLined
Input $\epsilon>0$; $a>0$; $b>0$ ; $\gamma>0$\\
Initialise $\bbeta^{(0)}=0$; $i=1$\\
\While {$||\bbeta^{(i)}-\bbeta^{(i-1)} || > \epsilon$ }{
 \For{$j \leftarrow 1$ \KwTo $P$}{
 $z_{j}=\bX_{j}^{T}(\bY-\bX_{-j}{\bbeta}^i_{-j})$\\
  $C_{2}=\sum_{k \neq j}|\beta_{j}|^{\frac{1}{2^{\gamma}}}+\frac{1}{b}$\\
 
  \uIf{$(\bX_{j}^{T}\bX_{j})^{-1}|z_{j}| \leq u(\beta_{-j})$}{$\beta_{j}^{(i+1)}=0$}
 \Else{
 $\beta_{j}^{(i+1,0)}=(\bX_{j}^{T}\bX_{j})^{-1}|z_{j}|$
 
 $k=1$
 
 \While{$|\beta_{j}^{(i, k)}-\beta_{j}^{(i, k-1)}|>\epsilon \quad \text{and} \quad \beta_{j}^{(i, k)}>0 \quad \text{and} \quad k\leq T$} {

  $\beta_{j}^{(i, k+1)}=(\bX_{j}^{T}\bX_{j})^{-1}|z_{j}|-(\bX_{j}^{T}\bX_{j})^{-1}\frac{C_{1}}{\beta^{(i, k)}_{j}+C_{2}(\beta^{(i, k)}_{j})^{1-\frac{1}{2^{\gamma}}}}$
  
  $k=k+1$
  }
  \uIf{$\beta_{j}^{(i, k)}<0 \quad \text{or} \quad k>T $}{
   ${\beta}^{(i+1)}_{j}=0$\;
   }
   \Else{${\beta}^{(i+1)}_{j}=\operatorname{sign}(z_{j})\beta_{j}^{(i,k)}$}
 }
 }
 $i=i+1$
 }
 \KwResult{${\bbeta}=\{\beta^{(i)}_1,\ldots, \beta^{(i)}_p\}$}
 \caption{Coordinate descent (CD) optimisation}\label{alg:VS}
\end{algorithm}

Further discussion on how to set $a$ and $b$ will be given in Section \ref{sec:oracle}.
The CD algorithm we propose here is not limited to the linear regression model, by using the proximal Newton map \citep{lee2014proximal} or the EM algorithm. Algorithm 2 can be extended to models with non-Gaussian likelihoods, see Section 4 of the Appendix.

\subsection{Convergence analysis}

In general, the convergence of the loss function alone cannot guarantee the convergence of $\left\{\bbeta^{i}\right\}_{i \geq 1}$. \cite{mazumder2011sparsenet} showed the convergence of CD algorithms for a subclass of non-convex non-smooth penalties for which the first and second-order derivatives are bounded. The authors also observed that for the type of log penalty proposed by \cite{friedman2012fast}, the CD algorithm can produce multiple limiting points without converging. 

The derivatives of the Bridge penalty and NSB are unbounded at zero, suggesting a local solution at zero, this causes discontinuity in the induced threshold operators. 
Recently, \cite{marjanovic2014l_} proved convergence of the CD algorithm for the Bridge penalty.
Here, we show that the CD algorithm with the NSB penalty also converges. We provide convergence analysis in the next two theorems. Details of the proofs can be found in Section 6 of the Appendix. 

\begin{theorem}\label{thm:convergence}
Suppose that given any fixed value of $\beta$ and the observation $(\bY,\bX)$, the log-likelihood term is bounded and the sequence $\left\{\bbeta^{i}\right\}_{i \geq 1}$ is generated by 
$\beta_{j}^{(i+1)}=T(\bbeta_{-j}^{(i)},\beta_{j}^{(i)})$  where the $T(\cdot)$ map is given by Equation (\ref{eq:KKT}) and  $\bbeta_{-j}^{(i)}=(\beta_{1}^{(i+1)},...,\beta_{j-1}^{(i+1)},\beta_{j+1}^{(i)},...,\beta_{p}^{(i)})$, then
$\|\bbeta^{i}-\bbeta^{i+1}\| \rightarrow{0} \mbox{ as } i \rightarrow \infty.$
\end{theorem}

\begin{theorem}\label{thm:local_minimizer}
Let $\bbeta^{\infty}$ be the estimator obtained by the CD Algorithm \ref{alg:VS}. Suppose $\bX$ satisfies (A4) with $\|\bbeta^{\infty}\|_{0} \leq \widetilde{p}$, then $\bbeta^{\infty}$ is a strict local minimizer of $L(\cdot)$. In other words, for any ${\bf e}=(e_{1},e_{2},...,e_{p})\in \mathrm{R^{P}}$ and $\|{\bf e}\|=1$, we have

$$\lim _{\alpha \downarrow 0+} \inf _{\alpha}\left\{\frac{L\left(\boldsymbol{\beta}^{\infty}+\alpha \boldsymbol{e}\right)-L\left(\boldsymbol{\beta}^{\infty}\right)}{\alpha}\right\} > 0.$$
\end{theorem}

The boundness condition in theorem  \ref{thm:convergence} is obviously true for the linear regression model considered here. Theorem \ref{thm:local_minimizer} then tells us that with the locally invertible assumption of the covariates over a small subspace, the convergent point is not only a stationary point but also a strict local minimizer. Empirically, we observe that the convergence of the solution paths is pleasingly well-behaved. Figure \ref{fig:3} demonstrates solution paths of the CD algorithm with two different initialisation strategies on a simulated sample of high dimensional regression problem with 990 zero entries and 10 nonzero entries in $\bbeta$. The design matrix $\bX$ is generated with pairwise correlation between $x_{i}$ and $x_{j}$ equal to $ 0.5^{|i-j|}$. Both the response and predictors have been standardized, so $\sigma_{0}^{2}=1$.  Data is simulated with $n=100$ and $p=1000$. From Figure \ref{fig:3}, we see that two solution paths finally converged to the same local minima with only two mistakes (they failed to identify one nonzero element and exclude one zero element).

\begin{figure}[h]
\includegraphics[width=\textwidth]{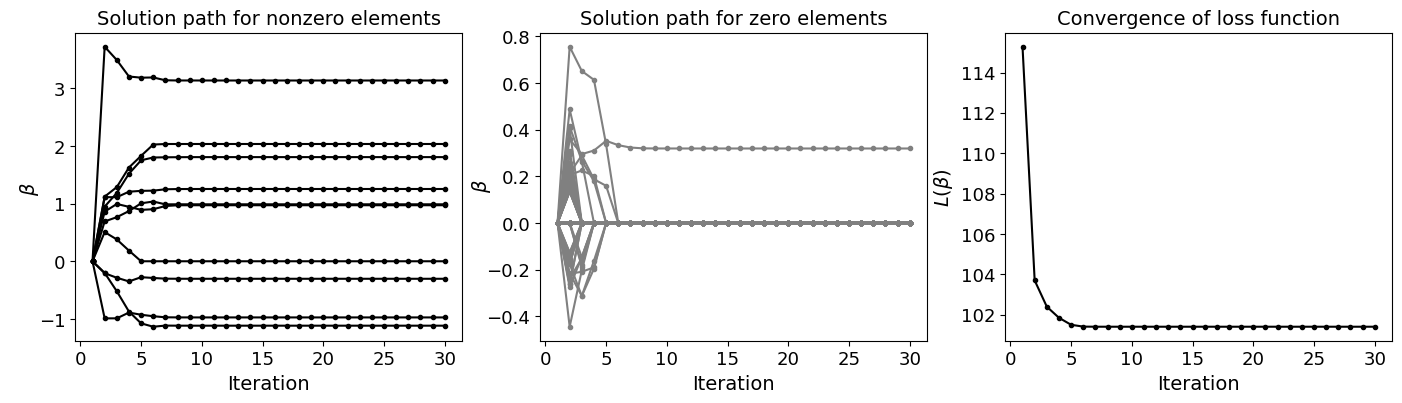}\\
\includegraphics[width=\textwidth]{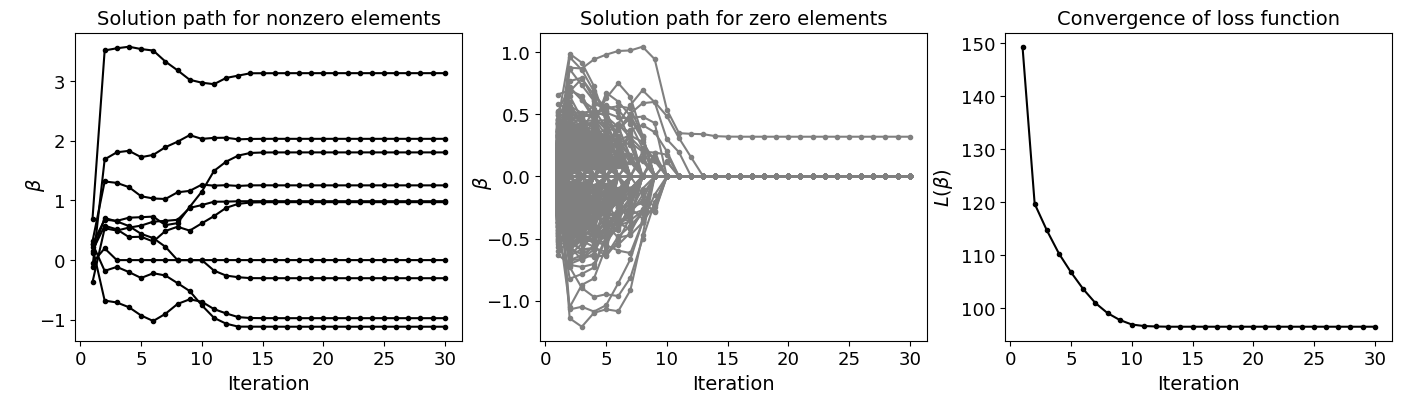}
\caption{Solution paths according to two initialisation strategies, $\bbeta_{initial}=0$ (top row) and random starting values $\bbeta_{initial} \sim N_{p}(0,I_{p})$ (bottom row). The left panel shows the solution paths for the 10 nonzero elements of $\bbeta$, the middle panel shows the solution paths for the (990) zero elements of $\bbeta$ and the right panel shows the path for the loss function $L(\bbeta)$. }\label{fig:3}
\end{figure}

\subsection{Oracle properties}
\label{sec:oracle}

In this section, we study the statistical properties of NSB. We will show that using an arbitrary, small positive number for $a$ and setting $b = C\frac{\log p_{n}}{p_{n}}$ for some positive constant $C$, model consistency 
and asymptotic normality of $\hat{\sigma}^{2}$ can be achieved. We state some extra assumptions here:\\
{\bf A6.} $\sum_{j \in S_{0}}|\beta_{0j}|^{\frac{1}{2^{\gamma}}}=O(s_{0})$.\\
{\bf A7.} There exist a constant $0<c<\infty$ such that $P\left(\left\|\frac{\bX_{S_{0}^{c}}^{T}\bX_{S_{0}}}{n}\right\|_{2,\infty} \leq c \right) \rightarrow 1$ as $n \rightarrow {\infty}$.\\
{\bf A8.} $\min_{j\in S_{0}}|\beta_{0,j}| \geq \max \left\{3\sigma_{0}\sqrt{\frac{2\log s_{0}}{v_{1}\log p_{n}}},3\sigma_{0}\left(\frac{\log p_{n}}{n}\right)^{\frac{1}{2-\alpha}}\right\}$ where $\alpha \in (\frac{1}{2^{\gamma}},2)$.\\
{\bf A9.} $\min_{j\in S_{0}}|\beta_{0,j}| \geq \max \left\{3\sigma_{0}\sqrt{\frac{2\log s_{0}}{v_{1}\log p_{n}}},3\sigma_{0}\left(\frac{ p_{n}}{s_{0}n}\right)^{\frac{1}{2-\alpha}}\right\}$ where $\alpha\in  (\frac{1}{2^{\gamma}},2)$.

Assumption (A6) specifies the strength of the overall true signal. Assumption (A7) states that the covariance between the relevant and irrelevant predictors should not be too strong. In addition, Assumptions (A8) and (A9) explicitly provide the weakest signal rate.  We see that the assumptions for oracle properties of posterior mode are stronger than the assumptions required for posterior consistency.

\begin{theorem}\label{thm:oracle}
Suppose assumptions (A1)-(A4), (A6) to (A7) are satisfied. For any small positive number $a$, under the following two cases, \\
(1): $\lim_{n \rightarrow{\infty}}\frac{p_{n}}{ns_{0}}=0$, $b\geq 1$ or $b \propto \frac{\log p_{n}}{p_{n}}$ and assumption (A8) holds.\\ 
(2): $\lim_{n \rightarrow{\infty}}\frac{p_{n}}{ns_{0}}=\infty$, $b \propto\frac{\log p_{n}}{p_{n}}$ and assumption (A9) holds.\\
we have:\\
(a) {\bf Model consistency} There is a strictly local minimizer $\hat{\bbeta}$  such that
$$\# \left\{j:\hat{\beta_{j}}\neq 0\right\}=S_{0}$$
with probability tending to 1 and, \\\\
(b) {\bf Asymptotic normality} With the local estimator $\hat{\bbeta}$, the correponding variance estimator $\hat{\sigma}^{2}$ from Equation (\ref{eq:sigma2}) has the property that
$$(\hat{\sigma}^{2}-\sigma_{0}^{2})\sqrt{n} \rightarrow N(0,\sigma_{0}^{4}E[\epsilon^{4}]-\sigma_{0}^{4}).$$
\end{theorem}

Proof of the theorem can be found in Section 7 of the Appendix. In practice, we also found that choosing any (positive, small) value for $a$ has no impact on the algorithm's performance. By default, we set $a=\frac{1}{2}$ and $b=C \frac{ \log p}{p}$. For linear regressions, we have used $C=1.5$ as the default setting. For both quantile regression and logistic regression, we set $C=0.5$ as the default setting. The potential improvement can be obtained if we tune it around the default value.

\section{Simulation and real data study}\label{sec:sim}

\subsection{Simulated data analysis }\label{sec:ess}
We first benchmark the PCG sampler for $\Lhalf$ prior with $\gamma=1,2$ against the Gibbs sampler scheme of the horseshoe prior \citep{makalic2015simple} in the high dimension sparse regression setting. 

We simulated two sets of data with moderate correlation between the covariates, ($\rho_{i,j}=0.5^{|i-j|}$) and high ($\rho_{i,j}=0.8^{|i-j|}$) correlation for the design matrix $\bX$ with $p=1000$ variables and sample size $n=100$. The true vector $\bbeta_{0}$ is constructed by assigning 10 nonzero elements $(3,1.5,2,1,1,0.5,-0.5,2,-1.2,-1)$ to locations $(1,2,5,10,13,19,26,31,46,51)$ and setting all the remaining coefficients to zero. We set the variance of the error to 1 ($\sigma_{0}^{2}=1$).

For each dataset, we ran 10 MCMC chains.  Each chain generated 10,000 samples after an initial burn-in of 10,000 samples. We combined the final 10,000 samples of all 10 chains to obtain a total of 100,000 samples. We assigned the Half-Cauchy prior $C^{+}(0,1)$ to the global shrinkage parameter of horseshoe prior as proposed by \cite{carvalho2010horseshoe}. For $\Lhalf$ prior with $\gamma=1$, we assigned $\lambda \sim \mathrm{Gamma}(100,1)$ and for $\gamma=2$, we assigned $\lambda \sim \mathrm{Gamma}(1, 200)$.

We calculated the
effective sample sizes (ESS) based on the formula from \cite{gelman2013bayesian}.  In Table \ref{tab:Table 1}, we report the average ESS over $\beta_{j}$ across $p=1000$ dimensions and the real computational time from a Macbook Air with M2 chip and 8 cores. As expected, the PCG sampler for $\Lhalf$ at $\gamma=1$ is faster than the PCG sampler for $\gamma=2$  due to the extra condition posterior of $v_{2}$ we need to sample with $\gamma=2$ (see Algorithm 1). 
In terms of effective sample size, the PCG sampler for $\Lhalf$ prior for both $\gamma=1$ and $\gamma=2$ significantly out performs the Gibbs sampler for the horseshoe prior, this remains the case when computational time is also accounted for (ESS/T).

\begin{table}[hbt!]
\centering
\begin{tabular}{|c|l|c|c|c|} 
\hline
\multicolumn{2}{|c|}{}       &\multicolumn{2}{c|}{\bf PCG}       &{\bf Gibbs }                \\ 
\multicolumn{2}{|c|}{}         &\multicolumn{1}{c}{($L_{\frac{1}{2}},\gamma=1$)}   &\multicolumn{1}{c|}{($L_{\frac{1}{2}},\gamma=2$)}   &{\bf (horseshoe)}   
  \\ 
\hline
\multicolumn{2}{|c|}{}           & \multicolumn{3}{|c|}{$\rho_{i,j}=0.5^{|i-j|}$}         \\ 
\hline
\multicolumn{2}{|c|}{\bf Time (Sec)}  &{\bf 281 }&372 &368 \\
\hline
\multirow{2}{*}{$\beta_j=0$}  &{\bf ESS}    & {\bf 67809}  &59985  & 37145            \\ 
&{\bf  ESS/T}     & {\bf 241}  &161  & 101   \\
\hline
\multirow{2}{*}{$\beta_j\neq 0$} 
&{\bf ESS} & {\bf 20355}  & 15878  & 7485              \\ 
&{\bf  ESS/T}  &{\bf 72} &43 &20  \\
\hline
\multirow{2}{*}{All $\beta_j$}
&{\bf ESS}   &{\bf 67334} & 59544  & 35848     \\
&{\bf  ESS/T}  &{\bf 240} &161 &97  \\
\hline
  \multicolumn{2}{|c|}{}                    & \multicolumn{3}{c|}{$\rho_{i,j}=0.8^{|i-j|}$}      \\ 
\hline
\multicolumn{2}{|c|}{\bf Time (Sec)}  &{\bf 300} &390 &364 \\
\hline
\multirow{2}{*}{$\beta_j=0$} &{\bf ESS}    & {\bf 66323}  & 60204   & 30531        \\ 
&{\bf  ESS/T}  &{\bf 221} &154 &84  \\
\hline
\multirow{2}{*}{$\beta_j\neq 0$}&{\bf ESS}   & 13483         & {\bf 15890}  & 6720         \\
&{\bf  ESS/T}  &{\bf 45} &41 &18  \\
\hline
\multirow{2}{*}{All $\beta_j$} &{\bf ESS}   & {\bf 65794 }        & 59761 & 30293         \\
&{\bf  ESS/T}  &{\bf 219} &153 &83  \\
\hline
\end{tabular}
\caption{
Averaged computing time (in seconds), based on real computation time from a MacBook Air with M2 chip and  8 core, over 100,000 samples. Effective sample size (ESS), and the ratio of ESS over time are shown for $L_{\frac{1}{2}}$ prior with PCG sampler ($\gamma=1,2)$, the horseshoe with Gibbs sampler, and two correlation structures $\rho_{i,j}$.}\label{tab:Table 1}
\end{table}

We now assess the statistical performance of $\Lhalf$ prior ($\gamma=1,3$) and NSB penalty ($\gamma=1,3$) under controlled conditions. We make comparisons with the full Bayesian horseshoe prior \citep{carvalho2010horseshoe}, the spike-and-slab g-prior (SSG) \citep{smith1996nonparametric,yang2016computational}, the non-separable spike-and-slab LASSO (NSSL) \citep{rovckova2018spike} and the MC+ penalty \citep{zhang2010nearly}.

Similarly to the previous section, we construct the AR(1) correlation structure for design matrix $\bX$ with correlation $\rho_{ij}=0.5^{|i-j|}$. We first test the performance of the models in high dimensional settings under the following scenarios when the number of nonzero regression coefficients is $s_{0}=10$.
$$
\begin{aligned}
& (1)\,n=500, p=1000, \sigma_{0}^{2}=1 \quad(2)\,n=500, p=1000, \sigma_{0}^{2}=3\\
& (3)\,n=100,p=1000,\sigma_{0}^{2}=1 \quad  (4)\,n=100, p=1000, \sigma_{0}^{2}=3.
\end{aligned}
$$
Then we test all zero scenarios with $s_{0}=0$.
$$(1)\ n=100,p=1000, \sigma_{0}^{2}=1 \quad(2)\ n=100,p=1000, \sigma_{0}^{2}=3.$$

For each of the above scenarios, we repeated the experiment 100 times. Each time, the true vector $\bbeta_{0}$ is constructed by assigning 10 non zero elements $(3,1.5,2,1,1,0.5,-0.5,2,-1.2,-1)$ to random locations and setting all the remaining $\bbeta$s to zero. 

For $L_{\frac{1}{2}}$ prior and the horseshoe prior, we used the same hyperprior setting as in Table \ref{tab:Table 1}. We considered the posterior means of $\bbeta$ and $\sigma^{2}$ as the point estimator, variable selection was performed by checking whether the marginal posterior credible intervals of each $\beta_j$ contained zero \citep{van2016horseshoe,wei2020contraction,tadesse2021handbook}. We set $95\%$ as the nominal level. Note that the choice of level is ad hoc, and different choices will lead to different variables being selected. Therefore our results in this regard serve only as an approximate guide.

For SSG, we set $\kappa_{0}=\kappa_{1}=1$. We sampled the posterior of the model using the Metropolis-Hastings with locally informed and thresholded proposal \citep{zhou2022dimension}. We also used the Rao-Blackwellized estimator to estimate $\beta_j$ as suggested by \cite{zhou2022dimension}. We sampled $\sigma^{2}$ from $\pi(\sigma^{2} \mid \mathcal{M}, \bY)$ and calculated the corresponding posterior mean. Variables selection was carried out based on posterior inclusion probability being larger than 0.5 \citep{barbieri2004optimal,barbieri2021median}. For all the Bayesian approaches, we ran MCMC for 20,000 iterations discarding the first 10,000 as burn in.

For the optimisation procedures, we used the default hyperparameter settings as discussed in Section \ref{sec:oracle} and the default settings for NSSL as in \cite{rovckova2018spike}. For MC+ penalty  \citep{zhang2010nearly}, we applied the Sparsenet algorithm \citep{mazumder2011sparsenet}, which also performs cross validation to set the values of tuning parameters $(\lambda,\gamma)$. For NSSL, the variance was estimated by the iterative algorithm from \cite{rovckova2018spike} and \cite{moran2018variance}. For NSB and MC+, we used the variance estimator from Equation (\ref{eq:sigma2}). 

We report simulation results by calculating the averages of the following criteria from 100 experiments: $L_{2}$ error (root mean square error), $L_{1}$ error (mean absolute error), FDR ($\%$ false discovery rate), FNDR ($\%$ false non-discovery rate), HD (Hamming distance), $\hat{\sigma}^{2}$ (variance estimator) and empirical posterior coverage probabilities. To compute the  coverage probability, we find the proportion of true $\bbeta$s falling within the estimated 95\% posterior credible interval, averaged separately for the true $\bbeta=0$ and true $\bbeta\neq 0$. Tables \ref{table:2} -  \ref{table:3} contain results for $s_{0}=10$ and
Table \ref{table:4} contains results for the all zero coefficients scenario ($s_{0}=0$). In the latter case, the FNDR is always zero and the FDR is either 0 or $100\%$ for any single experiment. Thus, we excluded these two quantities from Table \ref{table:4}.

Overall, among all the Bayesian procedures, the $\Lhalf$ prior with $\gamma=1$ and horseshoe prior have comparable performances across all criteria. The $\Lhalf$ prior with $\gamma=2$ only  slightly outperforms $\gamma=1$ at $n=500$ cases. For the other cases, it underselected the variables.
Although it is not shown in the table, our experiments have found that when $p>n$ and the sample size $n$ is small, the $L_{\frac{1}{2}}$ prior with $\gamma>2$ has a tendency to shrink everything to zero. Thus we don't suggest using them in this setting. In terms of variable selection, the SSG prior did slightly better than $\Lhalf$ prior with $\gamma=1$ and horseshoe prior in all the scenarios except in $n=100, p=1000,s_{0}=10,\sigma_{0}^{2}=3$ case.

\begin{table}[]
\centering
\scalebox{0.65}{
\begin{tabular}{|l|cccc|cccc|} 
\hline
                          & \multicolumn{4}{c|}{{\bf $L_{\frac{1}{2}}$ regression}}                                                     &  \multicolumn{1}{c|}{\bf horseshoe} & \multicolumn{1}{c|}{\bf SSG}      & \multicolumn{1}{c|}{\bf NSSL}   & \multicolumn{1}{c|}{\bf MC+} \\
                        
\hline
               & {\bf PCG ($\gamma=1$)}           & {\bf PCG ($\gamma=2$)}  & {\bf NSB ($\gamma=1$)}          & {\bf NSB ($\gamma=3$)}        &\multicolumn{4}{c|}{}              \\
                        
\hline
\multicolumn{9}{|c|}{ $n=500$, $p=1000$, $s_{0}=10$, $\sigma_{0}^{2}=1$}       \\ 
\hline
{\bf L2}                 & 0.4(0.3)      & 0.3(0.0)    & 0.2(0.0)                        & {\bf 0.1(0.0) }                      & 0.3(0.0)                     & 0.4(0.1)                                       & {\bf 0.1(0.0)  }              & 0.2(0.0)            \\ 
\hline
{\bf L1}            & 6.0(0.6)           & 5.4(0.5)      & {\bf 0.4(0.1)}                        & {\bf 0.4(0.0)}                      & 2.9(0.9)                     &  1.2(0.2)                                       &{\bf 0.4(0.1)   }         \    &  0.6(0.2)            \\ 
\hline
{\bf FDR $(\%)$}          & {\bf 0.0(0.0)}           & {\bf 0.0(0.0)}         &{\bf 0.0(0.0)}                           & {\bf 0.0(0.0)}                         & 15.3(2.9)                     & {\bf 0.0(0.0)}                                         & {\bf 0.0(0.0) }              & 23.0(4.1)               \\ 
\hline
{\bf FNDR $(\%)$}        &   0.2(0.1)           & 0.1(0.0)                           & {\bf 0.0(0.0)}                           & {\bf 0.0(0.0)}                  & {\bf 0.0(0.0) }                       & {\bf 0.0(0.0)  }                      & {\bf 0.0(0.0)}               &{\bf 0.0(0.0)  }          \\ 
\hline
{\bf HD }                 & 0.9(0.3)     &  0.6(0.1)      & {\bf 0.0(0.0) }                       & {\bf 0.0(0.0)}                     & 1.8(0.5)                       & {\bf 0.0(0.0)}                                          & {\bf 0.0(0.0)}         & 3.1(0.5)            \\ 
\hline
{\bf $\%$Cov }($\beta_j=0$)   &  99.8(0.6)  & {\bf 100(0.0)}    & NA          & NA                           &  99.8(0.4)                     & NA                                                  & NA                        & NA                    \\ 
\hline
{\bf $\%$Cov }($\beta_j \neq0$) & 80.0(3.6)  & {\bf 94.1(0.3)}    & NA                                & NA                          &  89.8(1.2)                     & NA                                                   & NA                        & NA                    \\ 
\hline
$\boldsymbol{\hat{\sigma}^2}$    &  0.9(0.2)      & 1.2(0.1)   & {\bf 1.0(0.1)}                        & {\bf 1.0(0.1) }                                & 0.8(0.0)                     &  0.9(0.1)                              &{\bf 1.0(0.1) }      
         & {\bf 1.0(0.2)}            \\ 
\hline
\multicolumn{9}{|c|}{$n=500$, $p=1000$, $s_{0}=10$, $\sigma_{0}^{2}=3$}                                                                     \\ 
\hline
{\bf L2}                 &  0.6(0.1)     & 0.7(0.3)    & 0.5(0.1)                        &  0.4(0.1)                       & 0.5(0.1)                     & 0.8(0.2)                                       & {\bf 0.3(0.1) }             & {\bf 0.3(0.1)  }          \\ 
\hline
{\bf L1 }             & 10.0(1.0)         & 10.9(1.0)    & 1.2(0.3)                        & { 0.8(0.3)}                      & 4.6(1.3)                     & 3.4(0.8)                                       & {\bf 0.7(0.2)  }             &  1.0(0.3)            \\ 
\hline
{\bf FDR $(\%)$}          &   1.0(0.1)          &{\bf 0.0(0.0)}   &{\bf 0.0(0.0)}                       & { 0.7(0.2) }                              & 15.3(2.3)                     & {\bf 0.0(0.0) }                                       &{\bf 0.0(0.0) }               &  10.8(4.1)            \\ 
\hline
{\bf FNDR $(\%)$}           & {\bf 0.0(0.0)}          & 0.1(0.0)   &  0.1(0.0)                        &{\bf 0.0(0.0) }                       & {\bf  0.0(0.0) }                           & {\bf  0.0(0.0) }                                           & {\bf  0.0(0.0) }           
      & {\bf  0.0(0.0) }           \\ 
\hline
{\bf HD}               &    0.4(0.3)      & 0.8(0.1)   & 0.3(0.5)                       & 0.3(0.2)       & 1.8(0.3)                       & {\bf 0.0 (0.0)  }                                    & 0.1(0.0)              & 2.9(0.5)            \\ 
\hline
{\bf $\%$Cov }($\beta_j=0$) & 99.8(0.2)   & {\bf 100(0.0)}     & NA                                & NA                          & 99.8(0.1)                     & NA                                                 & NA                        & NA                    \\ 
\hline
{\bf $\%$Cov }($\beta_j \neq0$) &87.5 (4.8)  &  {\bf 92.1(3.1)}   & NA                                & NA                              &  91.5(2.2)                     & NA                                               & NA                        & NA                    \\ 
\hline
$\boldsymbol{\hat{\sigma}^2}$     & 2.2(0.6)    & 3.3(0.2)   & 3.1(0.2)                        &{\bf 3.0(0.2)    }                           & 2.4(0.2)                     & 2.7(0.1)                                & {\bf 3.0(0.1)  }          \    &{\bf  3.0(0.1) }           \\
\hline
\end{tabular}
}
\caption{Comparison of performance under the setting $n=500, p=1000$ and the sparse case $s_0=10$, with two separate cases $\sigma_0^2=1,3$, using $\Lhalf$, NSB, 
horseshoe, SSG,
NSSL and MC+. Coverage probabilities are computed at the 95\% level. Results are based on averaging 100 repetitions of simulated data, respective standard derivations are shown in the brackets. We use bold to highlight the best performance.} \label{table:2}
\end{table}

\begin{table}[hbt!]
\scalebox{0.75}{
\centering\scalebox{0.90}{
\begin{tabular}{|l|cccc|cccc|} 
\hline
                          & \multicolumn{4}{c|}{{\bf $L_{\frac{1}{2}}$ regression}}                                                     &  \multicolumn{1}{c|}{\bf horseshoe} & \multicolumn{1}{c|}{\bf SSG}      & \multicolumn{1}{c|}{\bf NSSL}   & \multicolumn{1}{c|}{\bf MC+}\\ 
\hline
                          & {\bf PCG ($\gamma=1$)} 
                          & {\bf PCG ($\gamma=2$)}
                          & {\bf NSB ($\gamma=1$)}          & {\bf NSB ($\gamma=3$)}          & \multicolumn{4}{c|}{}              \\
                        
\hline
\multicolumn{9}{|c|}{ $n=100$, $p=1000$, $s_{0}=10$, $\sigma_{0}^{2}=1$}       \\ 
\hline
{\bf L2}                         & { 0.8(0.2) } 
&{2.67(0.3)}
& { 0.9(0.2)  }                      &  0.7(0.2)
& {0.8(0.2)}                     & 1.1(0.2)  
& 0.7(0.2)  & {\bf 0.6(0.2)}\\ 
\hline
{\bf L1}                         & 9.0(2.6)  
&{25.39(3.4)}
& 2.3(0.6)                        &  1.7(0.4)
& 9.4(2.1)                     & 4.2(0.3) 
&{\bf 1.6(0.4)}  &1.8(0.4)\\ 
\hline
{\bf FDR ($\%$) }                      & 10.0(0.1)
&{\bf 0.0(0.0)}  & 0.5(0.1) & 0.5(0.1) &15.9(1.1) &{\bf 0.0(0.0)} & 0.2(0.0)  &16.6(4.0)\\ 
\hline
{\bf FNDR ($\%$)}                      &  0.1(0.0)  
&{0.8(0.1)}
&  0.1(0.0)                        &  0.1(0.0)               
& {\bf 0.0(0.0) }                          &  0.1(0.0)
& 0.1(0.0)   &{\bf 0.0(0.0})\\ 
\hline
{\bf HD }                        & 2.0(0.1) 
&{8.3(0.4)}
     & 1.7(0.4)                        & 1.6(0.1) 
     & 1.9(0.2)                     & {\bf 1.4(0.2) }                                  &  1.5(0.1)                & 4.7(0.6)   \\ 
\hline
{\bf $\%$Cov}($\beta_j=0$)     &  99.9(0.0) 
&{\bf 100(0.0)}
& NA                                & NA   
&  99.8(0.2)                     & NA                                                & NA                         & NA           \\ 
\hline
{\bf $\%$Cov}($\beta_j \neq 0$) &   80.1(10.2) 
&{53.4(15.6)}
     & NA                                & NA    
     & {\bf 85.2(2.1)}                     & NA                                         &NA                         & NA           \\ 
\hline
$\boldsymbol{\hat{\sigma}^2}$           & 0.3(0.0) 
&{\bf 0.9(0.1)}
            & 1.5(0.2)                        & 1.4(0.1)  
            & 0.5(0.0)                     
            &{ 0.8(0.1)}                         & {\bf 1.1(0.1)}                  &{\bf 1.1(0.1)}   \\ 
\hline
\multicolumn{9}{|c|}{$n=100$, $p=1000$, $s_{0}=10$, $\sigma_{0}^{2}=3$}                                                                                                                                                                          \\ 
\hline
{\bf L2}                         &  1.5(0.5) 
&{2.9(0.3)}
      & { 1.3(0.4)}                        & {\bf 1.3(0.2) }        & 1.6(0.2)                     & 3.4(0.9)              & 1.4(0.2)        &{\bf 1.3(0.2)}   \\ 
\hline
{\bf L1  }                       & 18.9(3.3)  
&{30.1(5.6)}
     & 4.0(1.9)                        &  3.6(1.4) 
     & 13.6(2.4)                     & 13.2(5.6)                                & {\bf 3.5(0.6)  }              & 4.5(0.6)   \\ 
\hline
{\bf FDR ($\%$)}                        & 12.4(0.4)
&{\bf 0.0(0.0)}
       & 20.0(3.1)                        & 8.2(0.2)  
       & 16.2(0.5)                    & {\bf 0.0(0.0)}                                    &  0.8(0.1)                & 36.9(8.2)   \\ 
\hline
{\bf FNDR ($\%$)}                       &  0.4(0.1)   
&{0.9(0.1)}
      & 0.2(0.1)                        & 0.3(0.1)        & {\bf 0.0(0.0)}                           & 0.5(0.1)                          & 0.3(0.1)                &0.2(0.0)   \\ 
\hline
{\bf HD }                        & 4.0(1.0)  
&{8.3(0.4)}
      & 3.2(1.6)                        & 4.5(0.8)  
      &{\bf  1.9(0.3)}                     & 5.0(1.2)                               & 3.3(0.7)                & 11.6(1.7)  \\ 
\hline
{\bf $\%$Cov}($\beta_j=0$)     &  99.9(0.0) 
&{\bf 100(0.0)}
& NA                                & NA       
& 99.8(0.00)                       & NA                                         & NA                        & NA           \\ 
\hline
{\bf $\%$Cov}($\beta_j \neq 0$) & 62.1(4.7)  
&{50(3.9)}
         & NA                                & NA     & {\bf 76.3(3.1) }                    & NA                                      & NA                       & NA           \\ 
\hline
$\boldsymbol{\hat{\sigma}^2}$          & 0.7(0.1)
&{1.7(0.2)}
& 3.5(0.5)                     & { 3.2(0.4)}                & 2.5(0.3)                        & {\bf 3.0(0.2)  }                      & 3.7(0.4)                &  2.9(0.2)   \\
\hline
\end{tabular}
}}
\caption{Comparison of performance under the setting $n=100, p=1000$ and the sparse case $s_0=10$, with two separate cases $\sigma_0^2=1,3$, using $L_{\frac{1}{2}}$,  NSB,  horseshoe, SSG, NSSL and MC+. Coverage probabilities are computed at the 95\% level. Results are based on averaging 100 repetitions of simulated data, respective standard derivations are shown in the brackets. We use bold to highlight the best performances. }\label{table:3}
\end{table}

\begin{table}[hbt!]
\scalebox{0.75}{
\centering\scalebox{0.90}{
\begin{tabular}{|l|cccc|cccc|} 
\hline
                          & \multicolumn{4}{c|}{{\bf $L_{\frac{1}{2}}$ regression}}                                                     &  \multicolumn{1}{c|}{\bf horseshoe} & \multicolumn{1}{c|}{\bf SSG}      & \multicolumn{1}{c|}{\bf NSSL}   & \multicolumn{1}{c|}{\bf MC+}\\ 
\hline
                          & {\bf PCG ($\gamma=1$)} 
                          & {\bf PCG ($\gamma=2$)}
                          & {\bf NSB ($\gamma=1$)}          & {\bf NSB ($\gamma=3$)}          & \multicolumn{4}{c|}{}              \\
                        
\hline
\multicolumn{9}{|c|}{ $n=100$, $p=1000$, $s_{0}=0$, $\sigma_{0}^{2}=1$}       \\ 
\hline
{\bf L2}                         &  0.1(0.0)  
& 0.2(0.1)
& {\bf 0.0(0.0)  }                      &  {\bf 0.0(0.0)}
&  0.1(0.0)                     & 0.3(0.1)  
&0.7(0.1)  &   0.1(0.0)\\ 
\hline
{\bf L1}                         &  0.2(0.1)  
&{5.2(0.3)}
&{\bf 0.0(0.0)  }                      & {\bf 0.0(0.0)}
& 0.9(0.1)                     & 0.8(0.4) 
& 2.8(0.3)  &0.4(0.0)\\ 
\hline
{\bf HD }                        & {\bf 0.0(0.0)} 
&{\bf 0.0(0.0)}
     &{\bf 0.0(0.0)}                        &{\bf 0.0(0.0)} 
     &{\bf 0.0(0.0)}                     &  {\bf 0.0(0.0)}                                  
     & 13.4(1.2)                & 5.8(0.5)   \\ 
\hline
{\bf $\%$Cov}($\beta_j=0$)     & {\bf 100(0.0)} 
&{\bf 100(0.0)}
& NA                                & NA   
& {\bf 100(0.0)}                    & NA                                                & NA                         & NA           \\ 
\hline
$\boldsymbol{\hat{\sigma}^2}$           & 0.8(0.1) 
&{ 0.3(0.1)}
            & {\bf 1.0(0.1) }                       & {\bf 1.0(0.1)  }
            &  0.9(0.1)                     
            & 0.8(0.1)                         & 0.6(0.0)                  & 0.9(0.0)   \\ 
\hline
\multicolumn{9}{|c|}{$n=100$, $p=1000$, $s_{0}=0$, $\sigma_{0}^{2}=3$}                                                                                                                                                                          \\ 
\hline
{\bf L2}                         & {\bf 0.1(0.0)} 
&{0.3(0.1)}
      & { \bf 0.1(0.0)}                        &   0.1(0.0)         & {\bf 0.1(0.0) }                    & 0.2(0.1)              &{ 1.0(0.2)}        & 0.2(0.0)   \\ 
\hline
{\bf L1  }                       &{\bf 0.1(0.0)}  
     & 7.9(1.1)
     &  0.4(0.0)                        &  0.2(0.1) 
     & 1.2(0.1)                     & 1.1(0.3)                                        & 3.7(0.3)              & 0.7(0.1)   \\ 
\hline
{\bf HD }                        &{\bf 0.0(0.0) } 
&{\bf 0.0(0.0)}
      & 0.4(0.2)                        & 0.4(0.2)  
      & {\bf 0.0(0.0)  }                   & {\bf 0.0(0.0)}                                        & 7.9(0.4)                & 7.6(0.4)  \\ 
\hline
{\bf $\%$Cov}($\beta_j=0$)     & {\bf 100(0.0)} 
&{\bf 100(0.0)}
& NA                                & NA       
& {\bf 100(0.0) }                      & NA                                         & NA                        & NA           \\ 
\hline
$\boldsymbol{\hat{\sigma}^2}$          & 2.7(0.2)
&{1.7(0.2)}
& 2.8(0.4)                     & { 2.9(0.4)}                & 2.8(0.1)                        & {\bf 3.0(0.2)  }                      & 1.2(0.2)                &  2.6(0.1)   \\
\hline
\end{tabular}
}}
\caption{Comparison of performance under the setting $n=100, p=1000$ and $s_0=0$, with two separate cases $\sigma_0^2=1,3$, using $L_{\frac{1}{2}}$, NSB, horseshoe, SSG, NSSL and MC+. Coverage probabilities are computed at the 95\% level. Results are based on averaging 100 repetitions of simulated data, respective standard derivations are shown in the brackets. We use bold to highlight the best performance. }\label{table:4}
\end{table}

One remarkable thing is the underestimation of the variance by both $L_{\frac{1}{2}}$ prior and horseshoe priors when $p>n$. This is because there exists spurious sample correlation in the noise of predictors or between predictor and the response \citep{fan2008sure}. As a result, the realized noises are explained by the model with extra irrelevant variables, leading to underestimation of the variance $\sigma^2$ \citep{fan2012variance}. Table \ref{table:2} also shows that when the sample size becomes larger, this downward bias becomes smaller. In Section 8.1 of the appendix, we show that the variance estimation from the Bayesian approach significantly improves in large-scale problems ($n=2000$, $p=20000$, $s_{0}=10$).

The performances of optimisation procedures are generally comparable with those of the full Bayesian methods, except for $L_{2}$ and $L_{1}$ errors where they tend to perform better.  The results for NSB at $\gamma=1$ and $\gamma=3$ show that their performances are competitive and $\gamma=3$ tend to outperform $\gamma=1$.

\subsection{Real data analysis}

In this section, we analyse a real dataset from \cite{khan2016rafp} for the prediction of anti-freeze proteins (AFP). AFPs are a class of proteins produced by certain organisms to prevent ice formation or growth within their bodies in extremely cold conditions. Applications of AFP include food preservation, human cryopreservation and cryosurgery improvement, boosting freeze tolerance, ice and yoghurt formation \citep{breton2000biotechnological}.
Predicting the presence of AFPs in different species or discovering new AFPs is a field of interest in biotechnology. 

\cite{khan2016rafp} segmented protein sequence into two sub-sequences and calculated the occurrence frequency of amino acid compositions (AAC) and dipeptide compositions (DC) in each sequence separately. DC refers to the two connected amino acids in a protein sequence, which explores partial local information.  With 20 standard AACs, there are $20\times 20$ possible DC, therefore the data has $2\times(20+20\times 20)=840$ features in total. The data is partitioned into training and testing sets. The training set has a balanced setting with 300 AFPs and 300 non-AFPs, while the test set is highly imbalanced, with 181 AFPs and 9193 non-AFPs.

We fitted logistic regression and compared the results using $L_{\frac{1}{2}}$ prior ($\gamma=1$ and $\gamma=2$), NSB ($\gamma=1$ and $\gamma=3$) and horseshoe. We excluded SSG, NSSL and MC+ here as these cannot be readily adapted to the logistic regression setting. We set $\lambda \sim \mathrm{Gamma}(25,1)$ and $\lambda \sim \mathrm{Gamma}(1,170)$ respectively for $L_{\frac{1}{2}}$ prior with $\gamma=1$ and $\gamma=2$. We set $b= 0.004$ and $b=0.003$ respectively for NSB with $\gamma=1$ and $\gamma=3$. For the horseshoe prior, we set its hyperprior to Half-Cauchy $C^{+}(0,1)$. We are interested in whether the methods can correctly classify AFP while at the same time keeping the false positive and false negative small. We ran the MCMC for 20,000 iterations discarding the first 10,000 as burn-in and used the posterior predictive distribution for prediction.

We compare the results using accuracy, sensitivity, specificity and Youden’s index \citep{schisterman2005optimal,powers2020evaluation} with the following formula 
$$
\begin{aligned}
& \text{Accuracy}=\frac{TP+TN}{TP+TN+FN+FP}, \quad \text{Sensitivity}=\frac{TP}{TP+FN}, \\
& \text{Specificity}=\frac{TN}{FP+TN},    \quad \text{Youden’s index}=\text{Sensitivity}+\text{Specificity}-1.
\end{aligned}
$$
where TP is the number of true positives, TN is the number of true negatives, FN is the number of false negatives and FP is the number of false positives.
Table \ref{tb:AFP} shows the comparisons in terms of these metrics under the different methods.  We see that $L_{\frac{1}{2}}$ prior with $\gamma=1$ has the best performance in terms of all the criteria in Table \ref{tb:AFP}. 
  
To study the properties of AFP, it is also desirable to select relevant features, especially for AAC. For the continuous shrinkage priors, the significant features were identified by checking whether the $95\%$ posterior credible intervals contain zero. The significant AAC resulting from these methods are shown in Table \ref{tb:ACC_features}.  

Some of the significant AACs in Table \ref{tb:ACC_features} have biological justifications. For example, alanine is rich in type 1 AFPs  \citep{patel2010structures}, cysteine is rich in type 2 and type 5 AFPs \citep{ng1992structure} and glutamine is rich in type 4 AFPs \citep{deng1997amino}. Some AFPs have many hydrophobic amino acids on their surfaces \citep{howard2011neutron} and proteins with high proline are less likely to have hydrophobic amino acids. We see that with $95\%$ posterior credible intervals, the $\Lhalf$ prior at $\gamma=1$ and NSB at $\gamma=1$ selected all these features. While applying $\Lhalf$ at $\gamma=2$  resulted in 4 features being selected and
applying the horseshoe resulted in 2 features being selected, 
however, none of those selected were AACs.

\begin{table}
\centering
\begin{tabular}{|l|ccccc|} 
\hline
\multicolumn{1}{|l|}{} & {\bf $L_\frac{1}{2} (\gamma=1)$} & {\bf $L_\frac{1}{2} (\gamma=2)$}  & {\bf NSB($\gamma=1$)} & {\bf NSB($\gamma=3$)} & {\bf horseshoe} \\ 
\hline
Accuracy               & $\bf{85.5}\%$   & $85.3\%$          & $83.4\%$        & $82.7\%$     & $79.7\%$    \\ 
\hline
Sensitivity            & $\bf{90.6}\%$   & $85.6\%$          & $84.0\%$        & $78.5\%$      & $83.4\%$   \\ 
\hline
Specificity            & $\bf{85.4}\%$   & $85.3\%$         & $83.4\%$        & $82.8\%$      & $79.6\%$    \\ 
\hline
Youden's index         & \bf{0.76}            & 0.71                 & 0.67            & 0.61       & 0.63       \\ 
\hline
\end{tabular}
\caption{Comparison of the predictive performances of $L_{\frac{1}{2}} (\gamma=1,2)$, NSB $(\gamma=1,3)$ and horseshoe. The best performance method is in bold font for each metric.}\label{tb:AFP}
\end{table}

\begin{table}
\centering
\scalebox{0.80}{
\begin{tabular}{|l|cclclc|} 
\hline
Amino Acid Composition & $L_\frac{1}{2}$ & $L_{\frac{1}{4}}$    & {\bf horseshoe} & {\bf NSB($\gamma=1$)}      & \multicolumn{1}{c}{\bf NSB($\gamma=3$)} & {\bf Reference}              \\ 
\hline
Alanine                & $\checkmark$    &                      &           & $\checkmark$         & \multicolumn{1}{c}{$\checkmark$}    & \cite{patel2010structures}                    \\ 
\hline
Cysteine               & $\checkmark$    & \multicolumn{1}{l}{} &           & \multicolumn{1}{l}{} &                                     & \cite{ng1992structure}                     \\ 
\hline
Aspartic Acid          &                 & \multicolumn{1}{l}{} &           & $\checkmark$         &                                     & \multicolumn{1}{l|}{}  \\ 
\hline
Isoleucine             &                 & \multicolumn{1}{l}{} &           & $\checkmark$         & \multicolumn{1}{c}{$\checkmark$}    & \multicolumn{1}{l|}{}  \\ 
\hline
Leucine                &                 &                      &           & \multicolumn{1}{l}{} & \multicolumn{1}{c}{$\checkmark$}    & \multicolumn{1}{l|}{}  \\ 
\hline
Proline                & $\checkmark$    & \multicolumn{1}{l}{} &           & $\checkmark$         &                                     & \cite{howard2011neutron}                     \\ 
\hline
Glutamine              & $\checkmark$    &                      &           & $\checkmark$         &                                     & \cite{deng1997amino}                     \\ 
\hline
Arginine               &                 & \multicolumn{1}{l}{} &           & $\checkmark$         &                                     & \multicolumn{1}{l|}{}  \\ 
\hline
Valine                 & $\checkmark$    &                      &           & $\checkmark$         &                                     & \multicolumn{1}{l|}{}  \\
\hline
\end{tabular}
}
\caption{AAC features selected by $\Lhalf$ ($\gamma=1,2$), NSB ($\gamma=1,3$) and the horseshoe, as indicated by a tick. The last column shows the references that provide biological justifications of the corresponding AAC.}\label{tb:ACC_features}
\end{table}

\section{Conclusion and Discussion}\label{sec:conclusion}
In this article, we proposed a scale mixture of normal representation of the exponential power prior with $\alpha=(\frac{1}{2})^{\gamma}, \gamma=\{1,2,\ldots\}$. Based on this representation, we developed a {PCG} sampling scheme which is more efficient than traditional Gibbs sampling schemes and is scalable to high-dimensional problems. 

In addition,  inspired by the full Bayesian approach under the exponential power prior, we formulated a non-separable Bridge penalty. This is done by integrating over the hyperparameter $\lambda$ in the Bridge penalty with respect to gamma prior. The NSB penalty yields a combination of global-local adaptive shrinkage and threshold. With the Gaussian likelihood, its KKT condition lends itself to a fast optimisation algorithm via coordinate descent.
Since the NSB penalty is a non-convex, non-separable and non-Lipschitz function, we also provided a theoretical analysis to guarantee the convergence of the CD algorithm. Both the PCG sampler and CD optimisation algorithms are not limited to the Gaussian likelihoods, 
examples of extensions of our algorithm to more general settings can be found in Section 4 of the Appendix, where we gave some details for the case of the Bayesian logistic regression and Bayesian quantile regression with the asymmetric Laplace distribution. We also provide some additional numerical simulation results in Section 8.2 of the Appendix for the quantile regression example.

\section*{Supplementary Materials}

\begin{description}
\item[Appendix:] It contains all the proofs and extra numerical studies (appendix.pdf, PDF file).
\item[Source code:] It contains the PyTorch implementation of the algorithms and the protein data used in this article (code.zip, compressed file) and we refer the reader to the readme file therein for details. The codes are also available at \url{https://github.com/kexiongwen/Bayesian_L_half.git}.
\end{description}

\section*{Acknowledgments}
We would like to thank the editor, the associate editor, and the two reviewers for many valuable comments and helpful suggestions that led to an improved version of this article.

\section*{Disclosure Statement}

No potential conflict of interest was reported by the authors.

\bibliographystyle{apalike}
\bibliography{main.bib}

\begin{thebibliography}{}

\bibitem[Armagan, 2009]{armagan2009variational}
Armagan, A. (2009).
\newblock Variational bridge regression.
\newblock In {\em Artificial Intelligence and Statistics}, pages 17--24.

\bibitem[Barbieri and Berger, 2004]{barbieri2004optimal}
Barbieri, M.~M. and Berger, J.~O. (2004).
\newblock Optimal predictive model selection.

\bibitem[Barbieri et~al., 2021]{barbieri2021median}
Barbieri, M.~M., Berger, J.~O., George, E.~I., and Ro{\v{c}}kov{\'a}, V. (2021).
\newblock The median probability model and correlated variables.
\newblock {\em Bayesian Analysis}, 16(4):1085--1112.

\bibitem[Breton et~al., 2000]{breton2000biotechnological}
Breton, G., Danyluk, J., ois Ouellet, F., and Sarhan, F. (2000).
\newblock Biotechnological applications of plant freezing associated proteins.

\bibitem[B{\"u}hlmann and {v}an~{d}e Geer, 2011]{buhlmann2011statistics}
B{\"u}hlmann, P. and {v}an~{d}e Geer, S. (2011).
\newblock {\em Statistics for high-dimensional data: methods, theory and applications}.
\newblock Springer Science \& Business Media.

\bibitem[Carvalho et~al., 2010]{carvalho2010horseshoe}
Carvalho, C.~M., Polson, N.~G., and Scott, J.~G. (2010).
\newblock The horseshoe estimator for sparse signals.
\newblock {\em Biometrika}, 97(2):465--480.

\bibitem[Casella, 2001]{casella2001empirical}
Casella, G. (2001).
\newblock Empirical bayes gibbs sampling.
\newblock {\em Biostatistics}, 2(4):485--500.

\bibitem[Castillo et~al., 2015]{castillo2015bayesian}
Castillo, I., Schmidt-Hieber, J., {v}an~{d}er Vaart, A., et~al. (2015).
\newblock Bayesian linear regression with sparse priors.
\newblock {\em Annals of Statistics}, 43(5):1986--2018.

\bibitem[Datta and Ghosh, 2013]{datta2013asymptotic}
Datta, J. and Ghosh, J.~K. (2013).
\newblock Asymptotic properties of bayes risk for the horseshoe prior.

\bibitem[Deng et~al., 1997]{deng1997amino}
Deng, G., Andrews, D.~W., and Laursen, R.~A. (1997).
\newblock Amino acid sequence of a new type of antifreeze protein, from the longhorn sculpin myoxocephalus octodecimspinosis.
\newblock {\em FEBS letters}, 402(1):17--20.

\bibitem[Devroye, 2009]{devroye2009random}
Devroye, L. (2009).
\newblock Random variate generation for exponentially and polynomially tilted stable distributions.
\newblock {\em ACM Transactions on Modeling and Computer Simulation (TOMACS)}, 19(4):1--20.

\bibitem[Fan et~al., 2012]{fan2012variance}
Fan, J., Guo, S., and Hao, N. (2012).
\newblock Variance estimation using refitted cross-validation in ultrahigh dimensional regression.
\newblock {\em Journal of the Royal Statistical Society: Series B (Statistical Methodology)}, 74(1):37--65.

\bibitem[Fan and Lv, 2008]{fan2008sure}
Fan, J. and Lv, J. (2008).
\newblock Sure independence screening for ultrahigh dimensional feature space.
\newblock {\em Journal of the Royal Statistical Society: Series B (Statistical Methodology)}, 70(5):849--911.

\bibitem[Fan and Lv, 2011]{fan2011nonconcave}
Fan, J. and Lv, J. (2011).
\newblock Nonconcave penalized likelihood with np-dimensionality.
\newblock {\em IEEE Transactions on Information Theory}, 57(8):5467--5484.

\bibitem[Frank and Friedman, 1993]{frank1993statistical}
Frank, L.~E. and Friedman, J.~H. (1993).
\newblock A statistical view of some chemometrics regression tools.
\newblock {\em Technometrics}, 35(2):109--135.

\bibitem[Friedman, 2012]{friedman2012fast}
Friedman, J.~H. (2012).
\newblock Fast sparse regression and classification.
\newblock {\em International Journal of Forecasting}, 28(3):722--738.

\bibitem[Gelman et~al., 2013]{gelman2013bayesian}
Gelman, A., Carlin, J.~B., Stern, H.~S., Dunson, D.~B., Vehtari, A., and Rubin, D.~B. (2013).
\newblock {\em Bayesian data analysis}.
\newblock CRC press.

\bibitem[Gilks and Wild, 1992]{gilksWilks92}
Gilks, W.~R. and Wild, P. (1992).
\newblock Adaptive rejection sampling for gibbs sampling.
\newblock {\em Journal of the Royal Statistical Society. Series C (Applied Statistics)}, 41(2):337--348.

\bibitem[Howard et~al., 2011]{howard2011neutron}
Howard, E.~I., Blakeley, M.~P., Haertlein, M., Haertlein, I.~P., Mitschler, A., Fisher, S.~J., Siah, A.~C., Salvay, A.~G., Popov, A., Dieckmann, C.~M., et~al. (2011).
\newblock Neutron structure of type-iii antifreeze protein allows the reconstruction of afp--ice interface.
\newblock {\em Journal of Molecular Recognition}, 24(4):724--732.

\bibitem[Huang et~al., 2008]{huang2008asymptotic}
Huang, J., Horowitz, J.~L., Ma, S., et~al. (2008).
\newblock Asymptotic properties of bridge estimators in sparse high-dimensional regression models.
\newblock {\em The Annals of Statistics}, 36(2):587--613.

\bibitem[Khan et~al., 2016]{khan2016rafp}
Khan, S., Naseem, I., Togneri, R., and Bennamoun, M. (2016).
\newblock Rafp-pred: Robust prediction of antifreeze proteins using localized analysis of n-peptide compositions.
\newblock {\em IEEE/ACM Transactions on Computational Biology and Bioinformatics}, 15(1):244--250.

\bibitem[Knight et~al., 2000]{knight2000asymptotics}
Knight, K., Fu, W., et~al. (2000).
\newblock Asymptotics for lasso-type estimators.
\newblock {\em The Annals of statistics}, 28(5):1356--1378.

\bibitem[Lee et~al., 2014]{lee2014proximal}
Lee, J.~D., Sun, Y., and Saunders, M.~A. (2014).
\newblock Proximal newton-type methods for minimizing composite functions.
\newblock {\em SIAM Journal on Optimization}, 24(3):1420--1443.

\bibitem[Makalic and Schmidt, 2015]{makalic2015simple}
Makalic, E. and Schmidt, D.~F. (2015).
\newblock A simple sampler for the horseshoe estimator.
\newblock {\em IEEE Signal Processing Letters}, 23(1):179--182.

\bibitem[Mallick and Yi, 2018]{mallick2018bayesian}
Mallick, H. and Yi, N. (2018).
\newblock Bayesian bridge regression.
\newblock {\em Journal of applied statistics}, 45(6):988--1008.

\bibitem[Marjanovic and Solo, 2012]{marjanovic2012l_q}
Marjanovic, G. and Solo, V. (2012).
\newblock On $ l\_q $ optimization and matrix completion.
\newblock {\em IEEE Transactions on signal processing}, 60(11):5714--5724.

\bibitem[Marjanovic and Solo, 2013]{marjanovic2013exact}
Marjanovic, G. and Solo, V. (2013).
\newblock On exact l q denoising.
\newblock In {\em 2013 IEEE International Conference on Acoustics, Speech and Signal Processing}, pages 6068--6072. IEEE.

\bibitem[Marjanovic and Solo, 2014]{marjanovic2014l_}
Marjanovic, G. and Solo, V. (2014).
\newblock $ l\_ $\{$q$\}$ $ sparsity penalized linear regression with cyclic descent.
\newblock {\em IEEE Transactions on Signal Processing}, 62(6):1464--1475.

\bibitem[Mazumder et~al., 2011]{mazumder2011sparsenet}
Mazumder, R., Friedman, J.~H., and Hastie, T. (2011).
\newblock Sparsenet: Coordinate descent with nonconvex penalties.
\newblock {\em Journal of the American Statistical Association}, 106(495):1125--1138.

\bibitem[Moran et~al., 2018]{moran2018variance}
Moran, G.~E., Ro{\v{c}}kov{\'a}, V., George, E.~I., et~al. (2018).
\newblock Variance prior forms for high-dimensional bayesian variable selection.
\newblock {\em Bayesian Analysis}, pages 1091--1119.

\bibitem[Ng and Hew, 1992]{ng1992structure}
Ng, N. and Hew, C. (1992).
\newblock Structure of an antifreeze polypeptide from the sea raven. disulfide bonds and similarity to lectin-binding proteins.
\newblock {\em Journal of Biological Chemistry}, 267(23):16069--16075.

\bibitem[Park and {V}an Dyk, 2009]{park2009partially}
Park, T. and {V}an Dyk, D.~A. (2009).
\newblock Partially collapsed gibbs samplers: Illustrations and applications.
\newblock {\em Journal of Computational and Graphical Statistics}, 18(2):283--305.

\bibitem[Patel and Graether, 2010]{patel2010structures}
Patel, S.~N. and Graether, S.~P. (2010).
\newblock Structures and ice-binding faces of the alanine-rich type i antifreeze proteins.
\newblock {\em Biochemistry and Cell Biology}, 88(2):223--229.

\bibitem[Polson and Scott, 2012]{polson2012local}
Polson, N.~G. and Scott, J.~G. (2012).
\newblock Local shrinkage rules, l{\'e}vy processes and regularized regression.
\newblock {\em Journal of the Royal Statistical Society: Series B (Statistical Methodology)}, 74(2):287--311.

\bibitem[Polson et~al., 2014]{polson2014bayesian}
Polson, N.~G., Scott, J.~G., and Windle, J. (2014).
\newblock The bayesian bridge.
\newblock {\em Journal of the Royal Statistical Society: Series B (Statistical Methodology)}, 76(4):713--733.

\bibitem[Powers, 2020]{powers2020evaluation}
Powers, D.~M. (2020).
\newblock Evaluation: from precision, recall and f-measure to roc, informedness, markedness and correlation.
\newblock {\em arXiv preprint arXiv:2010.16061}.

\bibitem[Ro{\v{c}}kov{\'a} and George, 2016]{rovckova2016bayesian}
Ro{\v{c}}kov{\'a}, V. and George, E.~I. (2016).
\newblock Bayesian penalty mixing: the case of a non-separable penalty.
\newblock In {\em Statistical Analysis for High-Dimensional Data}, pages 233--254. Springer.

\bibitem[Ro{\v{c}}kov{\'a} and George, 2018]{rovckova2018spike}
Ro{\v{c}}kov{\'a}, V. and George, E.~I. (2018).
\newblock The spike-and-slab lasso.
\newblock {\em Journal of the American Statistical Association}, 113(521):431--444.

\bibitem[Schisterman et~al., 2005]{schisterman2005optimal}
Schisterman, E.~F., Perkins, N.~J., Liu, A., and Bondell, H. (2005).
\newblock Optimal cut-point and its corresponding youden index to discriminate individuals using pooled blood samples.
\newblock {\em Epidemiology}, pages 73--81.

\bibitem[Scott and Berger, 2010]{scott2010bayes}
Scott, J.~G. and Berger, J.~O. (2010).
\newblock Bayes and empirical-bayes multiplicity adjustment in the variable-selection problem.
\newblock {\em The Annals of Statistics}, pages 2587--2619.

\bibitem[Smith et~al., 1996]{smith1996nonparametric}
Smith, M., Kohn, R., et~al. (1996).
\newblock Nonparametric regression using bayesian variable selection.
\newblock {\em Journal of Econometrics}, 75(2):317--344.

\bibitem[Song and Liang, 2017]{song2017nearly}
Song, Q. and Liang, F. (2017).
\newblock Nearly optimal bayesian shrinkage for high dimensional regression.
\newblock {\em arXiv preprint arXiv:1712.08964}.

\bibitem[Tadesse and Vannucci, 2021]{tadesse2021handbook}
Tadesse, M. and Vannucci, M. (2021).
\newblock {\em Handbook of Bayesian Variable Selection}.
\newblock Chapman \& Hall/CRC Handbooks of Modern Statistical Methods. CRC Press.

\bibitem[{v}an~der Pas et~al., 2016]{van2016horseshoe}
{v}an~der Pas, S., Scott, J., Chakraborty, A., and Bhattacharya, A. (2016).
\newblock horseshoe: Implementation of the horseshoe prior.
\newblock {\em R package version 0.1. 0}, 12.

\bibitem[{V}an Dyk and Park, 2008]{van2008partially}
{V}an Dyk, D.~A. and Park, T. (2008).
\newblock Partially collapsed gibbs samplers: Theory and methods.
\newblock {\em Journal of the American Statistical Association}, 103(482):790--796.

\bibitem[Wei and Ghosal, 2020]{wei2020contraction}
Wei, R. and Ghosal, S. (2020).
\newblock Contraction properties of shrinkage priors in logistic regression.
\newblock {\em Journal of Statistical Planning and Inference}, 207:215--229.

\bibitem[West, 1987]{west1987scale}
West, M. (1987).
\newblock On scale mixtures of normal distributions.
\newblock {\em Biometrika}, 74(3):646--648.

\bibitem[Yang et~al., 2016]{yang2016computational}
Yang, Y., Wainwright, M.~J., and Jordan, M.~I. (2016).
\newblock On the computational complexity of high-dimensional bayesian variable selection.
\newblock {\em The Annals of Statistics}, 44(6):2497--2532.

\bibitem[Zhang et~al., 2010]{zhang2010nearly}
Zhang, C.-H. et~al. (2010).
\newblock Nearly unbiased variable selection under minimax concave penalty.
\newblock {\em The Annals of statistics}, 38(2):894--942.

\bibitem[Zhou et~al., 2022]{zhou2022dimension}
Zhou, Q., Yang, J., Vats, D., Roberts, G.~O., and Rosenthal, J.~S. (2022).
\newblock Dimension-free mixing for high-dimensional bayesian variable selection.
\newblock {\em Journal of the Royal Statistical Society Series B: Statistical Methodology}, 84(5):1751--1784.

\bibitem[Zou and Li, 2008]{zou2008one}
Zou, H. and Li, R. (2008).
\newblock One-step sparse estimates in nonconcave penalized likelihood models.
\newblock {\em Annals of statistics}, 36(4):1509.

\end{thebibliography}

\end{document}